\documentclass[sn-mathphys]{sn-jnl}

\usepackage{graphicx}
\usepackage{amssymb}
\usepackage{mathrsfs}

\def\eps{\epsilon}%
\def\tensor{\,\raise2pt\hbox{${}_{\otimes}$}\,}
\def\fdg{\,:\,}
\def\ptl{\partial}
\def\rest#1{\raise-2pt\hbox{${\lfloor_{#1}}$}}

\def\cal#1{\mathcal{#1}}

\def\larg#1{[#1]}
\def\ip#1#2{\langle#1,#2\rangle}

\def\grad{{\nabla}}
\newcommand{\leftexp}[2]{{\vphantom{#2}}^{#1}{#2}}

\def\halb{\frac{1}{2}}
\jyear{2021}%

\newtheorem{theorem}{Theorem}[section]
\newtheorem{lemma}[theorem]{Lemma}
\newtheorem{proposition}[theorem]{Proposition}
\newtheorem{corollary}[theorem]{Corollary}

\newtheorem{remark}[theorem]{Remark}
\newtheorem{definition}[theorem]{Definition}

\begin{document}

\title[Scattering for the Equivariant U(1) Problem]{Scattering for the Equivariant U(1) Problem}


\author*[1]{\fnm{Nishanth} \sur{Gudapati}}\email{ngudapati@holycross.edu}

\affil*[1]{\orgdiv{Department of Mathematics}, \orgname{College of the Holy Cross}, \orgaddress{\street{1 College Street}, \city{Worcester}, \postcode{01610}, \state{MA}, \country{United States}}}


\abstract{ Extending our previous works on the Cauchy problem for the $2+1$ equivariant Einstein-wave map system, we prove that the linear part dominates the nonlinear part of the wave maps equation coupled to the full set of the Einstein equations, for small data. A key ingredient in the proof is a nonlinear Morawetz estimate for the fully coupled equivariant Einstein-wave maps. The $2+1$ dimensional Einstein-wave map system occurs naturally in the $3+1$   vacuum Einstein equations of general relativity.}




\maketitle

\section{2+1 Equivariant Einstein-Wave Map System}
The aim of this work is to complete the task undertaken in \cite{BN_17}, on the scattering of the Cauchy problem of $2+1$ equivariant wave maps coupled to Einstein's equations for general relativity. Let $(M, g)$ be a 2+1 dimensional, globally hyperbolic, regular, Lorentzian spacetime what is foliated by 2 dimensional Riemannian hypersurfaces $(\Sigma, q)$ such that the group $SO(2)$ acts through isometries with a nonempty fixed point set $\Gamma$.  The fixed point set $\Gamma$ forms a timelike boundary of the quotient $M/ SO(2).$ In this setting, suppose $(M, g)$ can be represented in double null coordinates $(\xi, \eta, \theta)$ as 

\begin{align}
g =- e^{2Z} d\xi d\eta + r^2 d\theta^2 
\end{align}
where $Z = Z(\xi,\eta)$ and $r = r (\xi, \eta)$ are functions only of $(\xi, \eta)$, $\theta \in [0, 2\pi)$ is the parameter corresponding to $SO(2)$ action on $(\Sigma, q)$, $r(\xi,\eta) \in [0, \infty)$ is the radius of the area of group orbits of $SO(2)$ so that the curve $\{ r=0\}$ corresponds to the points on $\Gamma.$ $\xi$ and $\eta$ are the outgoing and incoming null coordinates respectively, such that 
\begin{align}
\xi = \eta, \quad \ptl_\xi r + \halb =0, \quad \ptl_\eta r - \halb =0 ,\quad \text{on}\quad \Gamma.
\end{align}
Now consider the 2+1 dimensional equivariant Einstein-wave map system

\begin{subequations}\label{eewm}
	\begin{align} 
	\mathbf{E}_{\mu \nu} =&\, \mathbf{T}_{\mu \nu} \label{eq:ee} \\
	\square_{g(u)} u=&\, \frac{\kappa^2 f_u(u) f(u)}{r^2}\label{eq:ewm},
	\end{align}
\end{subequations}
where $\mathbf{E}$ is the Einstein tensor of $(M, g)$ 
\[ \mathbf{E} \fdg= R_{\mu \nu} - \halb g_{\mu \nu} R_g, \] 
with the Ricci tensor $\mathbf{R}$ and the scalar curvature $R_g$ respectively. The tensor $\mathbf{T}_{\mu \nu}$ is the energy-momentum tensor associated to wave maps. It is given by
\begin{subequations}
    \begin{align}
\mathbf{T}_{\mu \nu}  &\fdg= \langle  \ptl_\mu U, \ptl_\nu U \rangle_{h(U)} - \halb g_{\mu \nu} \langle
\ptl^\sigma U, \ptl_\sigma U \rangle_{h(U)} \\
&=  h_{ij} \ptl_\mu U^i \ptl_\nu U^j - \halb g_{\mu \nu} g^{\alpha \beta}  h_{ij} \ptl_\alpha U^i \ptl_\beta U^j.
\end{align}
\end{subequations}
The map 
\begin{align} \label{map} 
U \fdg (M, g) \to (N, h)
\end{align}
is an equivariant wave map, where the target $(N, h)$ is a surface of revolution such that

\begin{align}
f(0) =0,\quad f'(0) =1 \quad \text{and} \quad \int^u_0 f(s) ds \to \infty \quad \text{as} \quad u \to \infty
\end{align}
with the odd and smooth generating function $f.$ The equivariant ansatz $U = (u, \kappa \theta)$, where $u$ is independent of the angular variable $\theta$, reduces the wave maps equation 
\[ \square_g\,  U^i + \leftexp{(h)} \Gamma^i_{jk} (U) g^{\mu \nu} \ptl_\mu  U^j\ptl_\nu U^k =0\]
to 
\begin{align}
\square_{g} u=&\, \frac{\kappa^2 f_u(u) f(u)}{r^2},
\end{align}

\noindent as it is well known. The constant $\kappa$ represents the degree of the equivariant wave map \eqref{map}. 
The precise value of the constant $\kappa$ is not relevant for the small data analysis in our paper. Thus, we set $\kappa=1$ for simplicity.  
In the double null coordinate system introduced above, the Ricci tensor $\mathbf{R}$ of $(M, g)$ is
\begin{align*}
\mathbf{R}_{\xi\xi} =& r^{-1} (2Z_\xi r_\xi - r_{\xi\xi}),\\
\mathbf{R}_{\xi\eta} =& -(2 Z_{\xi\eta} + r^{-1} r_{\xi \eta}),\\
\mathbf{R}_{\eta \eta} =& r^{-1} (2Z_\eta r_\eta - r_{\eta\eta}), \\
\mathbf{R}_{\theta\theta} =& 4 r\,e^{-2 Z}  r_{\xi\eta}, \\
\mathbf{R}_{\xi\theta} =& 0\,\,\text{and} \\
\mathbf{R}_{\eta\theta} =& 0.
\end{align*}
The scalar curvature is
\begin{align*}
R_g={}&8 e^{-2 Z} \bigl( \ptl_\xi \ptl_\eta Z + r^{-1}r_{\xi\eta} \bigr). 
\end{align*}

In the following we shall write down the Einstein tensor $ \mathbf{E}$ and the energy-momentum tensor $\mathbf{T}$ in the double null coordinates. Let us start with the Einstein tensor $\mathbf{E}:$

\begin{align*}
\mathbf{E}_{\xi\xi} =& r^{-1} (2 \ptl_\xi Z \, \ptl_\xi r- \ptl^2_\xi r),\\
\mathbf{E}_{\xi\eta} =& r^{-1} \ptl_ \xi \ptl_\eta\, r,\\
\mathbf{E}_{\eta \eta} =& r^{-1} (2 \ptl_\eta Z \ptl_\eta r - \ptl^2_\eta r),\\
\mathbf{E}_{\theta \theta} =& -4 r^2 \,e^{-2 Z} \ptl_\xi \ptl_\eta Z,\\
\mathbf{E}_{\xi \theta} =& 0\,\,\text{and} \\
\mathbf{E}_{\eta\theta} =& 0. \\
\intertext{In the computation of the energy-momentum tensor it will be helpful to compute the wave map Lagrangian $\mathcal{L}$, defined as,}
\mathcal{L} \fdg=& \ip{\ptl^\alpha U}{ \ptl_\alpha U}_{h(U)} = g^{\alpha \beta} h_{ij} (U) \ptl_{\alpha} U^i \ptl_{\beta} U^j  \\
=&-4 e^{-2Z} \ptl_\xi u \ptl_\eta u + \frac{f^2 (u)}{r^2}. 
\intertext{The components of $\mathbf{T}_{\mu \nu}$  are}
\mathbf{T}_{\xi \xi} =& \ptl_\xi u \, \ptl_\xi u,\\
\mathbf{T}_{\eta \eta} =& \ptl_\eta u \, \ptl_\eta u,  \\
\mathbf{T}_{\xi \eta} =& \frac{e^{2Z}}{4} \frac{f^2(u)}{r^2}, \\\ 
\mathbf{T}_{\theta\theta} =& \frac{r^2}{2} e^{-2Z}\left ( 4\ptl_\eta  u\ptl_\xi
u + e^{2Z}\frac{f^2(u)}{r^2} \right ). 
\end{align*} 
Suppose we introduce coordinates \[T = \frac{\xi + \eta}{2}, \quad R= \frac{\xi-\eta}{2}\]
\noindent we can construct the `retarded time' coordinates $(\eta, R, \theta)$ such that the metric $g$ of $M$ can be expressed as
\begin{align}
g = -e^{2Z (\eta, R)} d\eta^2 - 2 e^{2Z (\eta, R)} d\eta\, dR  + r^2 (\eta, R) d \theta^2.
\end{align}
Let us now calculate the relevant quantities in the retarded-time coordinates:

Ricci tensor $\mathbf{R}$
\begin{align*}
\mathbf{R}_{\eta\eta} =& r^{-1} \left( r \ptl^2_R Z - \ptl_R r \ptl_R Z + 2 \ptl_\eta Z \ptl_\eta r -2r \ptl^2 _{\eta R} Z - \ptl_\eta^2 r) \right) ,\\
\mathbf{R}_{\eta R} =& r^{-1}\left(\ptl_R Z \ptl_R r + r \ptl^2_R Z - 2r\ptl^2_{\eta R} Z - \ptl^2_{\eta R} r \right),\\
\mathbf{R}_{R R} =&r^{-1} \left( 2 \ptl_R Z \ptl_R r - \ptl^2 r \right) , \\
\mathbf{R}_{\theta\theta} =&  r e^{-2Z} (-\ptl^2_R r + 2 \ptl^2_{\eta R} r), \\
\mathbf{R}_{\eta\theta} =& 0\,\,\text{and} \\
\mathbf{R}_{R\theta} =& 0.
\end{align*}
The scalar curvature can be represented as
\begin{align}
R_g = -r^{-1}  2 e^{-2Z}( r \ptl^2_R Z + \ptl^2_R r -2r \ptl^2_{\eta R} Z - 2 \ptl^2_{\eta R}\, r )
\end{align}

Consequently, the Einstein tensor $\mathbf{E}$ is given by 
\begin{subequations}
	\label{bondi-ee}
	\begin{align}
	\mathbf{E}_{\eta\eta} =&  -r^{-1} ( -\ptl_R Z \ptl_R r + \ptl_R^2 r - \ptl^2_{R \eta} r )\notag \\
	&-r^{-1} \left( \ptl_\eta r \ptl_R Z + \ptl_R Z \ptl_\eta r - 2\ptl_\eta Z \ptl_\eta r - \ptl_{\eta R}^2 r + \ptl^2_\eta r   \right),\\
	\mathbf{E}_{\eta R} =&  - r^{-1}\left( - \ptl_R Z \ptl_R r + \ptl^2_R r - \ptl^2_{\eta R} r        \right) ,\\
	\mathbf{E}_{RR} =& r^{-1} \left(  2 \ptl_R Z \ptl_R r - \ptl_R^2 r \right),\\
	\mathbf{E}_{\theta \theta} =& e^{-2Z} r^2 \left( \ptl_R^2 Z - 2 \ptl^2_{\eta R} Z \right),\\
	\mathbf{E}_{R \theta} =& 0\,\,\text{and} \\
	\mathbf{E}_{\eta\theta} =& 0.
	\end{align}
\end{subequations}
\noindent The energy-momentum tensor $\mathbf{T}$ is given by

\begin{align}
\mathbf{T}_{\mu \nu}  \fdg= \langle  \ptl_\mu U, \ptl_\nu U \rangle_h - \halb g_{\mu \nu} \mathcal{L}.
\end{align}

\begin{align}
\mathcal{L} =& h_{ij} g^{\mu \nu} \ptl_\mu U^i \ptl_\nu U^j \notag \\
=& -e^{2Z} (\ptl_\eta u)^2 -2 e^{-2Z} \ptl_\eta u \, \ptl_R u + \frac{f^2(u)}{r^2}
\end{align}
\begin{subequations} 
	\label{bondi-se}
	\begin{align}
	\mathbf{T}_{\eta \eta} =& \halb \left(  (\ptl_\eta u)^2 -2 \ptl_\eta u \ptl_R u + e^{2Z} \frac{f^2 (u)}{r^2}\right),\\
	\mathbf{T}_{\eta R} =&\halb \left( - (\ptl_\eta u)^2 + \frac{f^2 (u)}{r^2} \right)  ,  \\
	\mathbf{T}_{RR} =& (\ptl_R u)^2  ,\\ 
	\mathbf{T}_{\theta\theta} =& \frac{r^2}{2} e^{-2Z}\left ( (\ptl_\eta u)^2 + 2 \ptl_\eta u \ptl_R u + e^{2Z}\frac{f^2(u)}{r^2} \right ). 
	\end{align}
\end{subequations}
If we introduce a time coordinate function $T \fdg =  \eta + R$,
we can represent the metric $g$ also in terms of $(T, R, \theta)$ coordinates
\[ ds^2 = e^{2Z(T, R)} (-dT^2 + dR^2) + r^2(T, R) d\theta^2.\]
The Ricci tensor in $(T, R, \theta)$ is

\begin{align}
\mathbf{R}_{TT} =&r^{-1} \left(  \ptl_R Z \ptl_R r+ r \ptl^2_R Z + \ptl_T Z \ptl_T r - r \ptl_T^2 Z - \ptl^2_T r\right)  ,\\
\mathbf{R}_{TR} =& r^{-1}\left(   \ptl_Rr \ptl_T Z + \ptl_R Z \ptl_T r - \ptl^2_{TR} r     \right),\\
\mathbf{R}_{R R} =&-r^{-1} \left(   -\ptl_R Z \ptl_R r + r \ptl_R^2 Z + \ptl^2_R r - \ptl_T Z \ptl_T r - r\ptl^2_T Z \right) , \\
\mathbf{R}_{\theta\theta} =&  r e^{-2Z} \left( -\ptl^2_T r + \ptl^2_R r  \right), \\
\mathbf{R}_{T\theta} =& 0\,\,\text{and} \\
\mathbf{R}_{R\theta} =& 0.
\end{align}

The scalar curvature is given by

\[ R_g = -2r^{-1} e^{-2Z} \left(  r (-\ptl^2_T Z + \ptl^2_R Z) + (-\ptl^2_T r + \ptl^2_R r) \right) \]

The Einstein tensor
\begin{align}
\mathbf{E}_{TT} =& r^{-1} \left( \ptl_R Z \ptl_R r + \ptl^2_R r - \ptl_T Z \ptl_T r \right),\\
\mathbf{E}_{TR} =& r^{-1} \left(  \ptl_R r \ptl_T Z + \ptl_R Z \ptl_T r - \ptl^2_{TR} r\right),\\
\mathbf{E}_{R R} =& r^{-1} \left(  \ptl_R Z \ptl_R r + \ptl_T Z \ptl_T r - \ptl^2_T r \right),\\
\mathbf{E}_{\theta \theta} =&  r^2 \,e^{-2 Z} \left(  -\ptl^2_T Z + \ptl^2_R Z \right),\\
\mathbf{E}_{T \theta} =& 0\,\,\text{and} \\
\mathbf{E}_{R\theta} =& 0.
\end{align}

Furthermore, we have

\[\mathcal{L} = - e^{-2Z} (\ptl_T u)^2 + e^{-2Z} (\ptl_R u)^2 + \frac{f^2(u)}{r^2}\]

Thus,
\begin{align}
\mathbf{T}_{TT}=& \halb \left( (\ptl_T u)^2 + (\ptl_R u)^2 + e^{2Z} \frac{f^2(u)}{r^2}   \right),\\
\mathbf{T}_{TR}=& (\ptl_T u \ptl_R u) ,\\
\mathbf{T}_{RR}=& \halb \left( (\ptl_T u)^2 + (\ptl_R u)^2 - e^{2Z} \frac{f^2(u)}{r^2}   \right), \\
\mathbf{T}_{\theta \theta} =& \halb r^2 \left(  -e^{-2Z} (\ptl_T u)^2 + e^{-2Z} (\ptl_R u)^2 + \frac{f^2(u)}{r^2}       \right)
\end{align}

\subsection*{Initial Value Problem}
We shall set up the initial value problem at a $T=0$ Cauchy hypersurface $(\Sigma_0,q_0)$. The initial data of geometric part of the $2+1$ Einstein-wave map system consists of a Riemannian 2-manifold $(\Sigma_0, q_0)$ and a symmetric 2-tensor $k_0$. 
The Einstein equations for general relativity are an overdetermined system, thus to formulate an initial value problem  the following constraint equations are imposed on $(\Sigma_0, q_0, k_0)$ 
\begin{subequations}
	\begin{align}
	R_q + tr(k_0)^2 -\Vert k_0 \Vert^2_q =& 2 \mathbf{e} \label{hamc} \\
	\grad(q_0)_a (k_0)^a_b - \grad(q_0)_b (k_0)^a_a =&\mathbf{m}, \label{momc}
	\end{align}
\end{subequations}
where  $\mathbf{e}$
and $\mathbf{m}$ are the energy and momentum densities respectively:

\begin{subequations}
\begin{align}
\mathbf{e} =&\, \halb \left( \Big\Vert \mathcal{N} (U) \Big\Vert^2_h + \big\Vert \mathcal{S} (U) \Big\Vert^2_h + \Big\Vert \frac{1}{r} \ptl_\theta (U) \Big\Vert^2_h \right) \\
\mathbf{m} =&\,  \mathcal{N} (U) \cdot \mathcal{S}(U)
\end{align}
\end{subequations}
with $\mathcal{N} \fdg = e^{-Z}\ptl_T$ and $\mathcal{S} \fdg = e^{-Z}\ptl_R.$

It follows from a standard result that there exists a unique (upto an isometry) smooth, globally hyperbolic, future development $(M, g, U)$ of the initial data $(\Sigma_0)$ \cite{Bruhat_Geroch_classic}. Subsequently, $k_0$ is the second fundamental form of the embedding $\Sigma_0 \hookrightarrow M$ and if $\mathcal{N}$ is the unit normal of $\Sigma_0 \hookrightarrow M$ the above system is 

\begin{subequations}
	\begin{align}
	\mathbf{E} (\mathcal{N},\mathcal{N}) = \mathbf{T}(\mathcal{N}, \mathcal{N}) \\
	\mathbf{E} (\mathcal{N}, e_1) = \mathbf{T}(\mathcal{N}, e_1).
	\end{align}
\end{subequations}

The initial data of the wave maps equation are 
\begin{subequations}
	\begin{align}
	U_0 \fdg& (\Sigma_0, q_0) \to (N, h) \\
	U_1 \fdg& (\Sigma_0, q_0) \to T_{U_0} (N, h)
	\end{align}
\end{subequations}
Therefore the 5-tuple $(\Sigma_0, q_0, k_0, U_0,U_1)$ that satisfies the constraint equations constitutes the initial data of the Einstein-wave map system. For our convenience, we introduce the following quantities on the initial Cauchy surface
\begin{subequations} \label{initialdata}
	\begin{align}
	u\vert_{\Sigma_0} = u_0,&\quad \ptl_T u\vert_{\Sigma_0} =u_1, \\
	Z\vert_{\Sigma_0} =Z_0,& \quad  \ptl_T Z\vert_{\Sigma_0} =Z_1, \\
	r\vert_{\Sigma_0} =r_0, &\quad \ptl_T r \vert_{\Sigma_0} =r_1.
	\end{align}
\end{subequations}
with $r_1 \vert_{\Gamma} =0$ and $\ptl_R r_0 \vert_{\Gamma} =1,$ in the $(T, R, \theta)$ coordinate system. 

\begin{remark}\label{rem:u-v}
	Define a function $v$ such that $u = Rv$, then the following results hold
	\begin{enumerate}
		\item $v$ satisfies the 4+1 wave equation: 
		\begin{equation}\label{Problem1}
		\left. \begin{array}{rcl}
		\leftexp{4+1}{\square}\, v &=& F(v)\,\,\,\,\,\,\,\,\,\,\,\,\,\,\, \,\,\,\,\,\,\,\,\,\,\,\,\,\,\,\,\,\,\textnormal{on}\,\, \mathbb{R}^{4+1}\\
		v_0 = v (0, x) & \textnormal{and}& v_1 = \ptl_T v (0, x) \,\,\,\,\,\,\,\, \textnormal{on}\,\, \mathbb{R}^4\\\end{array} 
		\right\}
		\end{equation}
		where $\leftexp{4+1}{\square}$ is the wave operator in the $4+1$ dimensional Minkowski space and  
  
		\begin{align}
		F(v) =& \left(e^{2Z} - 1 + \left(\frac{r}{R}\partial_{\eta} r + \frac{1}{2}\right) - \left(\frac{r}{R}\partial_{\xi} r - \frac{1}{2}\right)\right) \frac{v}{r^{2}} \notag\\
		&+ 2 \ptl_\xi v \ptl_\eta \log (\frac{r}{R}) + 2 \ptl_\eta v \ptl_\xi \log (\frac{r}{R})+ e^{2Z} \frac{R^{2}}{r^{2}} v^{3} \zeta(R v), 
		\end{align}
   $\zeta$ is a smooth function (see below).
		\item If we define $E_0 \fdg = \int_{\Sigma_0} \mathbf{T} (\mathcal{N}, \mathcal{N}) \,\bar{\mu}_{q_0} = \int_{\Sigma_0} \mathbf{e} \,\bar{\mu}_{q_0}$, the energy of $v$ satisfies the estimate
		\begin{align}
		E_0 \geq& \Vert u_0\Vert^2_{\dot{H}^1(\mathbb{R}^2)} + \Vert u_1 \Vert^2_{L^2(\mathbb{R}^2)} \\
		\geq& \Vert v_0\Vert^2_{\dot{H}^1(\mathbb{R}^4)} + \Vert v_1 \Vert^2_{L^2(\mathbb{R}^4)}
		\end{align}  
	\end{enumerate} 
 $\bar{\mu}_{q_0}$ is the volume form of $(\Sigma_0, q_0).$
	The estimate follows from the Hardy's inequality and assumptions on the target manifold (see \cite{BN_17}).
\end{remark}
Let us briefly remark on the derivation of the wave equation \eqref{Problem1}. In null coordinates, the covariant wave operator $\square_g$ can be expressed as: 

\begin{align}
    \square_g = -2 e^{-2Z} r^{-1} (\ptl_\eta (r \ptl_\xi ) + \ptl_\xi ( r \ptl_\eta )).
\end{align}

Now, using  $u=Rv$, in $(T, R, \theta)$ coordinates we can rewrite  $\square_g u $ in terms of the $4+1$ dimensional Minkowski wave operator for $v$ i.e., $\leftexp{4+1}{\square} v$.  In view of our assumption that the function $f$ in the metric of the target manifold, is a smooth, odd function such that $f(0)=1$ and $f_u (0)=1, $ we have 

\begin{align}
    f(u)f_u (u) = u + u^3 \zeta( u)
    \intertext{for some smooth function $\zeta$. Consequently,}
    \frac{f_u(u) f(u)}{r^2}= \frac{u}{r^2} + \frac{\zeta (u) u^3}{r^2} = \frac{R}{r^2} V + \frac{\zeta (Rv) R^3}{r^2} V^3.
\end{align}
This allows to derive the equation \eqref{Problem1}. We refer the reader to Proposition 1.1 in \cite{BN_17} for more details on this derivation. 

In the following, we shall define the notion of asymptotic flatness of the 2-metric $(\Sigma, q)$

\begin{definition}[Asymptotic flatness and AV-mass]
	Suppose $(\Sigma, q, k)$ is a Riemannian 2-metric, then it is called asymptotically flat if 
	\begin{enumerate}
		\item for some compact set $\Omega$, $\exists$ a global  diffeomorphism $\Sigma \setminus \Omega \to \mathbb{R}^2\setminus \mathbb{B}_1(0)$
		\item The $(\Sigma, q, k)$ admits the asymptotic expansion 
  
		\begin{align}
		q_{ab} =& \rho^{-m_{AV}} (\delta_{ab} + \mathcal{O}(\rho^{-1})) \label{2d-spatial-decay} \\
		k=&\mathcal{O}(\rho^{-2}),\quad tr(k) = \mathcal{O}(\rho^{m_{AV}-3}) \\
		\intertext{ and the energy momentum tensor}
		\mathbf{T}_{\mu \nu}=& \mathcal{O} (\rho^{m_{AV}-3}) 
		\end{align}
		in the asymptotic region, where $\delta_{ab}$ is the Euclidean metric and $\rho = \vert x \vert$. We define the quantity $m_{AV}$ as the Ashtekar-Varadarajan mass of the Cauchy hypersurface $(\Sigma, q, k)$ \cite{ash_var}.
	\end{enumerate}
\end{definition}
The concept of geometric mass for an asymptotically flat manifold in general relativity is based on the decay properties of the spatial metric relative to a reference (Eulidean) metric in the asymptotic region. In the 2+1 dimensional case, the geometric mass is defined in terms of the Eucliean 2-metric $\delta_{ab}$. $\rho = \vert x \vert $ is a radial function in the asymptotic region and the decay rate is measured in terms of $\rho$ and  this in turn allows us to read-off the geometric mass $m_{AV}.$

The spatial decay properties  (eq. \eqref{2d-spatial-decay})  of the 2+1 dimensional spatial metric $q$ should be contrasted with the decay properties of the metric $(\bar{q})$ in the 3+1 dimensional case: 

\begin{align}
    \bar{q}_{ij} = \left( 1+ \frac{{m_{\text{ADM}}}}{r} \right) \bar{\delta}_{ij} + \mathcal{O} (r^{-2})
\end{align}
where $r = \vert x \vert,$ in the $3$ dimensional asymptotic region and $m_{\text{ADM}}$ is the well-known ADM mass of the metric $q.$  

Recall the form of the 2-metric for our problem:

\begin{align}
q = e^{2Z} dR^2 + r^2 d\theta^2, 
\end{align}
then a calculation shows that $m_{AV}$ is given by 
\begin{align}
m_{AV} = 2(1-e^{-\tilde{Z}_\infty}), \quad (\Sigma, q)
\intertext{i.e., the spatial metric $q$ can be written as}
q= \rho^{-2 (1-e^{\tilde{Z}_\infty} )} (d \rho^2 +  \rho^2 d \theta^2 ) 
\end{align}
in the spatial asymptotic region, where 
\begin{align}
e^{2\tilde{Z}_\infty} \fdg = \lim_{R \to \infty} e^{2\tilde{Z}},  \quad e^{2\tilde{Z}} \fdg = e^{2Z} \left(\frac{\ptl r}{\ptl R}\right)^{-2} \quad \text{on}\quad  (\Sigma, q).
\end{align}

\subsection*{Review of Known Results}
To put our results into context let first us review  the results for the 2+1 equivariant Einstein-wave map system and the  the literature on critical wave maps  shall be  discussed later.  
For the initial data described in the previous section, we have the following global existence result from 
\cite{AGS_15}, which also holds without the smallness restriction of the initial data in view of a non-concentration of energy argument \cite{diss}.
\begin{theorem}[Global regularity of equivariant Einstein-wave maps]\label{thm:main-first}
	Let $E_0 < \eps^2$ for $\eps $ sufficiently small and let
	$(M, g, u)$ be the maximal Cauchy development of an asymptotically flat, compactly supported, regular Cauchy data set for the $2+1$ equivariant Einstein-wave map problem \eqref{eewm} with target $(N, h)$ satisfying
	\begin{equation} \label{eq:nospherecond-first} 
	\int_0^s f(s') ds' \to \infty \quad \text{for} \quad s \to \infty.
	\end{equation} 
	Then $(M, g, u)$ is regular and causally geodesically complete.
\end{theorem}
As a part of this result, we established that the energy controls the metric components of the $2+1$ dimensional spacetime (see for example Corollary 5.10 and Corollary 5.11 in \cite{AGS_15}): 

\begin{align}
    \vert e^{2Z} -1 \vert, \quad   \Big \vert \frac{R}{r} -1 \Big \vert, \quad  \vert u \vert \leq E_0.
\end{align}

The results in \cite{diss_13, AGS_15, BN_17} and the current work are part of the larger program to understand the global behaviour of the $3+1$ Einstein equations with one translational $U(1)$ symmetry. 
In global existence proof of \cite{AGS_15}, it was proved that, in future development of the initial data of the 2+1 equivariant Einstein-wave map system, the trapped surfaces cannot form. In particular, it was proved that 

\begin{theorem} [No trapped surfaces; Theorem 3.1 in \cite{AGS_15}]
	Suppose $(M, g, U)$ is the regular, globally hyperbolic maximal development of the 2+1 equivariant Einstein-wave map system, define $\mathcal{Q} \fdg = M \setminus U(1)$ and 
	\begin{align}
	\mathcal{R} \fdg = \{ p \in \mathcal{Q} \quad \textnormal{such that} \quad \ptl_\xi r >0, \ptl_\eta r <0 \}
	\end{align}
	then, 
	\begin{enumerate}
		\item $\mathcal{Q} = \mathcal{R}$
		\item $q \in \mathcal{Q},$
		\begin{align}
		0 \leq m(q) < m_{\infty} <1
		\end{align}
		where $m_{\infty}$ is the spatial asymptotic limit of $m$ on $(\Sigma, q, k).$
	\end{enumerate}
\end{theorem}

Consider the system of coupled nonlinear wave equations occurring in the 2+1 Einstein-equivariant wave map system:

\begin{subequations}
	\begin{align}
	-\ptl^2_T Z + \ptl^2_R Z =& - e^{2Z} \frac{f^2(u)}{8r^2} - 
	\frac{1}{8}\big(-(\ptl_T u)^2 + (\ptl_R u)^2\big) \label{eq:waveZ1}\\
	-\ptl^2_T r + \ptl^2_R r=& r \kappa (\ptl_T u + \ptl_R u)^2\label{eq:waver}\\
	\square_{g(u)} u=&\frac{f(u) f_u (u)}{r^2}.
	\label{eq:waveuTR}
	\end{align}
\end{subequations}  
In a previous work, we have established scattering for the following partially coupled Einstein-wave map equations, classified as Problems I and II below.
\subsection*{Problem I}
Consider a function $v$ such that
\begin{equation}\label{Wave1}
\left. \begin{array}{rcl}
\leftexp{4+1}{\square}\, v &=& F(v)\,\,\,\,\,\,\,\,\,\,\,\,\,\,\, \,\,\,\,\,\,\,\,\,\,\,\,\,\,\,\,\,\,\textnormal{on}\,\, \mathbb{R}^{4+1}\\
v_0 = v (0, x) & \textnormal{and}& v_1 = \ptl_T v (0, x) \,\,\,\,\,\,\,\, \textnormal{on}\,\, \mathbb{R}^4\\\end{array} 
\right\}
\end{equation}
with 
\[F(v) = \left(e^{2Z} - 1 + \left(\frac{r}{R}\partial_{\eta} r + \frac{1}{2}\right) - \left(\frac{r}{R}\partial_{\xi} r - \frac{1}{2}\right)\right) \frac{v}{r^{2}} + e^{2Z} \frac{R^{2}}{r^{2}} v^{3} \zeta(R v), \]
where $Z$ and $r$ satisfy the equations

 \begin{align}
     -4 \ptl_\xi \ptl_\eta Z =& \frac{1}{2} \left(4 \ptl_\eta \ptl_\eta u + e^{2Z} \frac{f^2 (u)}{r^2} \right) \\
r^{-1} \ptl_\eta \ptl_ \eta r =& \frac{e^{2Z}}{4} \frac{f^2(u)}{r^2}, \quad u = Rv.
 \end{align}
 In comparison with the equation \eqref{Problem1}, it may be noted that the nonlinearity in equation \eqref{Wave1} in Problem I does not contain the $\displaystyle 2 \ptl_\eta v \ptl_\eta \log \frac{r}{R} + 2 \ptl_\eta v \ptl_\xi  \log \frac{r}{R}$ term. 

\subsection*{Problem II}
Suppose the function $v$ is such that 
\begin{equation} \label{Wave2}
\left. \begin{array}{rcl}
\leftexp{4+1}{\square}\, v &=& F(v)\,\,\,\,\,\,\,\,\,\,\,\,\,\,\, \,\,\,\,\,\,\,\,\,\,\,\,\,\,\,\,\,\,\textnormal{on}\,\, \mathbb{R}^{4+1}\\
v_0 = v (0, x) & \textnormal{and}& v_1 = \ptl_T v (0, x) \,\,\,\,\,\,\,\, \textnormal{on}\,\, \mathbb{R}^4\\\end{array} 
\right\}
\end{equation}
where 
\begin{align}
F(v) =& \left( \frac{1}{r} \ptl_\eta r + \frac{1}{2R} \right)\ptl_\xi v  + \left( \frac{1}{r} \ptl_\xi r - \frac{1}{2R} \right)\ptl_\eta v \notag\\
&\quad + \left( \left(\frac{r}{R} \ptl_\eta r + \halb \right) - \left(\frac{r}{R} \ptl_\xi r - \halb \right) \right) \frac{v}{r^2} \notag\\
&\quad + \frac{R^2}{r^2} v^3 \zeta(R v)
\end{align}
and $r$ satisfies the equation 

\begin{align}
   r^{-1} \ptl_\eta \ptl_ \eta r =& \frac{e^{2Z}}{4} \frac{f^2(u)}{r^2}, \quad u= Rv.
 \end{align}
 It may be noted that the equation \eqref{Wave2} is a special case of the equation \eqref{Problem1}, wherein we restricted to the special case $Z \equiv 0.$
We were able to prove scattering for both Problems I and II:
\begin{theorem}[Theorem 2.11 and Theorem 3.2 in \cite{BN_17}]
	Suppose $v$ is a global solution of the wave equation \eqref{Wave1} or  \eqref{Wave2}, with the
	energy $E_0  < \eps^2$, then $v$ scatters forward and backward in time
\end{theorem}

In our context, by scattering we mean that the solution to the given nonlinear equation converges to a solution of a linear equation in the energy topology, as $T \to \infty.$
The coupling of the wave maps equation with the Einstein equations introduces nonlocality into the problem, which causes significant difficulties in closing the estimates at critical regularity. 

It may be noted that the equations \eqref{Wave1} and \eqref{Wave2} in Problem I and II respectively, are special cases of the fully coupled nonlinear equivariant Einstein wave maps equation \eqref{Problem1}.


The methods employed in the wave maps equation are based on a optimal control of the Morawetz and Strichartz estimates that allow  us to control the bootsrap estimates at critical regularity. 
In this work we are able to prove scattering for the fully coupled wave maps equation:

\begin{equation} \label{wave-full}
\left. \begin{array}{rcl}
\leftexp{4+1}{\square}\, v &=& F(v)\,\,\,\,\,\,\,\,\,\,\,\,\,\,\, \,\,\,\,\,\,\,\,\,\,\,\,\,\,\,\,\,\,\textnormal{on}\,\, \mathbb{R}^{4+1}\\
v_0 = v (0, x) & \textnormal{and}& v_1 = \ptl_T v (0, x) \,\,\,\,\,\,\,\, \textnormal{on}\,\, \mathbb{R}^4\\\end{array} 
\right\}
\end{equation}

\begin{align}
F(v)=& \left( e^{2Z}-1 + (\frac{r}{R} \ptl_\eta r + \halb)
- (\frac{r}{R} \ptl_\xi r - \halb)  \right) \frac{v}{r^2}\notag\\ &+ 2 \ptl_\xi v \ptl_\eta \log \frac{r}{R} + 2 \ptl_\eta v \ptl_\xi \log (\frac{r}{R}) + e^{2Z} \frac{R^2}{r^2} v^3\zeta (Rv)
\end{align}

\begin{theorem}\label{scattering-full}
	Suppose $E_0 <\eps^2,$ for $\eps$ sufficiently small and $(M, g, U)$ is a global solution of the 2+1 equivariant Einstein-wave map system, then $v$ scatters backward and forward in time in the energy topology i.e., 
 $$\Vert v - v_S \Vert \to 0 $$ in the energy topology, as $T \to 0$,
 where $v_S$ is a solution of the linear equation $\leftexp{4+1}{\square} v_S =0.$
\end{theorem}

We would like to remind the reader that we are able to control the 
metric components in terms of energy f.g., 

$$\Big \vert \frac{R}{r} -1 \Big \vert, \quad  Z \leq E_0$$
so we need not make assumptions on the smallness of the metric components. 

In the proof of the above scattering theorem, we use the following nonlinear Morawetz estimate:
\begin{theorem}[Nonlinear Morawetz estimate]
	Suppose $(M, g, U)$ is a globally regular solution of the 2+1 Einstein-equivariant wave map system, then 
	\begin{align} \label{non-Mora}
	\int_{\mathbb{R}^{4+1}} \frac{v^2}{R^3} \bar{\mu}_{\check{g}} \leq \Vert v_0 \Vert^2_{\dot{H}^1 (\mathbb{R}^4)} + \Vert v_1 \Vert^2_{L^2(\mathbb{R}^4)}.
	\end{align}
	where $\bar{\mu}_{\check{g}}$ is the volume form of the $4+1$ Minkowski space $(\mathbb{R}^{4+1}, \check{g}),$ with
	\begin{align}
	\check{g} = -dT^2 + dR^2+ R^2 d\omega_{\mathbb{S}^3}.
	\end{align}
\end{theorem}  

Subsequently, using the scattering result for $v$, we obtain scattering for the fully coupled equivarant wave maps: 

\begin{theorem} [Temporal Scattering] \label{eewm-scatter}
	Let $E_0 < \eps^2$ for $\eps $ sufficiently small and let
	$(M, g, u)$ be the maximal geodesically complete Cauchy development of an asymptotically flat, compactly supported, regular Cauchy data set for the $2+1$ equivariant Einstein-wave map problem \eqref{eewm} with target $(N, h)$ satisfying
	\begin{equation} 
	\int_0^s f(s') ds' \to \infty \quad \text{when }  s \to \infty.
	\end{equation} 
	Then the maximal development $(M, g, U)$ is such that $U$ converges to a solution of the linearized (around the Minkowski space)  wave maps equation  \eqref{eq:ewm}  as $T \to \infty$ along future pointing the timelike curves. 
\end{theorem}
We would like to point out that, both conceptually and technically, the wave map fields (and their conjugate momenta) is the dominant field that represents the true dynamical degrees of freedom of the $2+1$ dimensional Einstein-wave map system. 

The key ingredient in the proof is the non-linear Morawetz estimate \eqref{non-Mora} for the fully coupled Einstein-wave map system. An analogous result holds for the partially coupled equivariant Einstein-wave system \eqref{Wave2}
(Problem II), but \eqref{non-Mora} now holds for the fully coupled system. Likewise,a similar result holds for scattering backwards in time. 

\noindent In view of the fact that the Jacobians of the transition functions between the $(T, R, \theta)$, double-null and retarded-time coordinates are constants, the scattering result of Theorem \ref{eewm-scatter} can be adapted to establish equivalent results along the $\eta= const.$ in the asymptotic region.  

\begin{corollary}
	Suppose $(M, g, U)$ is the regular, globally hyperbolic, causally geodesically complete maximal development of the asymptotically flat initial data $(\Sigma, q, k, U)$ with $u$ compactly supported,  then $u$ admits the decay rate
	\begin{align} 
	\vert u \vert \leq  (1+\xi)^{-1/2} (1+ \eta)^{-1/2}
	\end{align}
	in the temporal asymptotic region $(T \to \infty)$ and 
	\begin{align}
	u = \mathscr{O}(R^{-1/2})
	\end{align}
	as $R \to \infty$ in the $\eta=const.$ hypersurfaces. 
\end{corollary}
Let us now discuss a few salient points about the 2+1 dimensional equivariant Einstein-wave map system.
\subsubsection*{Metric Functions} The wave map $U$ and its conjugate momentum form the dynamical degrees of freedom of the system and drive the dynamics of the whole system. However, in our quest to understand the asymptotic geometry of the maximal future development $(M, g, U)$ of the system, the form of the Einstein tensor poses obstacles. For instance, consider the equation

\begin{equation}\label{1+1example}
\left. \begin{array}{rcl}
\leftexp{1+1}{\square}\, w &=&e^x\,\,\,\,\,\,\,\,\,\,\,\,\,\,\, \,\,\,\,\,\,\,\,\,\,\,\,\,\,\,\,\,\,\,\,\,\,\,\,\,\,\,\textnormal{on}\,\, \mathbb{R}^{1+1}\\
w (0, x) = 0 & \textnormal{and}& \ptl_T w (0, x)=0 \,\,\,\,\,\,\,\,\,\,\,\textnormal{on}\,\, \mathbb{R}\\\end{array} 
\right\}
\end{equation}

The energy of $w$ in the equation  \eqref{1+1example} goes out of control asymptotically as $T \to \infty$ even though we start with `small' initial data. Thus we are motivated to analyze the asymptotic behaviour separately for metric functions to rule out such behaviour. As per our analysis, we expect that, as the maximal development $(M, g, U)$ attains its scattering state near temporal infinity, it converges to the Minkowski space. In other words, as the maximal development attains its scattering state near temporal infinity, the geometric mass (Ashtekar-Varadarajan) converges to zero. This topic, with the together the behaviour of the curvature, merits a dedicated discussion and shall be addressed in a separate work.

\subsubsection*{Sharp Asymptotics}
We would like to remark that the scattering methods have an advantage in the characterization of the asymptotic geometry of the initial value problem of the system as the decay rate so obtained is sharp compared to the  energy methods in 2+1 dimensions: 

	\begin{align} \label{kubota}
	\vert u \vert \leq  (1+ \xi)^{-1/2} (1+ \eta)^{-1/2}\quad \text{in} \quad \mathbb{R}^{2+1},
	\end{align}

in view of the fact that the geometric mass and spatial asymptotics are conserved globally, including the region where the system $(M, g, U)$ attains its scattering state. 
\subsubsection*{The Small Data Problem and the Target}

We would also like to remark that our scattering result reaffirms the intuition the curvature of the target manifold $(N, h)$ does not play a decisive role in the scattering in the small data context.

Wave maps from $\mathbb{R}^{2+1}$ with symmetry have been studied in a  series of landmark papers of Christodoulou-Tahvildar-Zadeh and Shatah \cite{chris_tah1, chris_tah2, jal_tah, jal_tah1}, where global regularity and asymptotics of initial value problems of equivariant and spherically symmetric wave maps were studied (see also\cite{grillakis}). 

In \cite{hu_16} the Einstein-scalar wave system without symmetry  was studied and stability of Minkowski space in exponential time $(\sqrt{t}^{-1})$ was established. The techniques in \cite{hu_16} were inspired from \cite{LR_10} which provided a simpler proof of the seminal and classic work on stability of Minkowski space \cite{CK_94}. 
The asymptotic behaviour of perturbations of the 3+1 Minkowski space in harmonic coordinates is studied in \cite{HL_17}. 


In a pioneering and influential series of works \cite{tao_all}, global behaviour of wave maps in multiple dimensions and regularity levels has been studied.  Local and global well-posedness of wave maps at various regularity levels are studied in \cite{tat_besovl, tat_besovh, tat_isom}. Global regularity for small data for critical wave maps $U \fdg \mathbb{R}^{2+1} \to \mathbb{H}^2$ was proved in \cite{krieg_wmcrit}. In a voluminous work \cite{krieg_schlag_ccwm}, the important question of concentration compactness for wave maps $U \fdg \mathbb{R}^{2+1} \to \mathbb{H}^2$ was thoroughly settled. Likewise, scattering for large data for critical wave maps with more general targets was settled in \cite{sterb_tata_long}\cite{sterb_tata_main}.

\subsection*{Notation}
A wave operator $\leftexp{n+1}{\square}$ without a metric subscript refers to a flat-space wave operator in $n+1$ dimensions.  In $4+1$ dimensional Minkowski space, the  spacetime volume  is denoted as $\bar{\mu}_{\check{g}}$ and spatial volume is denoted as $\bar{\mu}_{\check{q}}$. For example,  in polar coordinates, $ \bar{\mu}_{\check{g}}= R^3  dT dR d \theta $  and $ \bar{\mu}_{\check{q}} = R^3 dR d \theta. $ In this work, we use the inequality symbol $\leq$ for estimates that hold up to a constant. 
We shall use the Einstein summation convention throughout.

\section{Scattering for 2+1 Equivariant Einstein Wave Map System}
Recall the 2+1 dimensional equivarant Einstein-wave system: 
\begin{align}
\mathbf{E}_{\mu \nu} =& \mathbf{T}_{\mu \nu}, \quad \text{on} \quad (M, g) \\
\square_{g(u)} u =& \frac{f_u(u)f(u)}{r^2}
\end{align}
for the initial data $(\Sigma_0, q_0, K_0, u_0, u_1).$  Now consider the following nonlinear second-order hyperbolic partial differential equation $(n \neq 1)$: 

\begin{equation}\label{scatteringexample}
\left. \begin{array}{rcl}
L\, \psi &=& N(\psi)\,\,\,\,\,\,\,\,\,\,\,\,\,\,\, \,\,\,\,\,\,\,\,\,\,\,\,\,\,\,\,\,\,\,\,\,\,\,\,\,\,\,\textnormal{on}\,\, (M, g)\\
\psi (0, x) = \psi_0 & \textnormal{and}& \ptl_T \psi (0, x)=\psi_1 \,\,\,\,\,\,\,\,\,\,\,\,\,\,\,\,\,\textnormal{on}\,\, (\Sigma, q)\\\end{array} 
\right\}
\end{equation}
where $L$ is a linear operator and $N$ is the nonlinearity.  We say that the wave equation \eqref{n+1example} exhibits temporal scattering if there exits a solution of a corresponding linearized equation $\psi_S$ and an energy topology $E_\psi$ such that 
\begin{equation}\label{scatteringstate}
\left. \begin{array}{rcl}
L\, \psi_S &=& 0\,\,\,\,\,\,\,\,\,\,\,\,\,\,\,\,\,\,\,\,\,\,\,\,\,\, \,\,\,\,\,\,\,\,\,\,\,\,\,\,\,\,\,\,\,\,\,\,\,\,\,\,\,\textnormal{on}\,\, (M, g)\\
\psi_S (0, x) = (\psi_S)_0 & \textnormal{and}& \ptl_T \psi (0, x)=(\psi_S)_1 \,\,\,\,\,\,\,\,\,\textnormal{on}\,\, (\Sigma, q)\\\end{array} 
\right\}
\end{equation}
and 
\begin{align}
\Vert \psi - \psi_S \Vert_{E_\psi} \to 0, \quad \text{as} \quad  T \to \infty,
\end{align}
where $((\psi_S)_0, (\psi_S)_1)$ are the scattering initial data of \eqref{scatteringstate}. As we already remarked, we shall recast the wave maps equation as a 4+1 dimensional wave equation,

\begin{equation}\label{n+1example}
\left. \begin{array}{rcl}
\leftexp{4+1}{\square}\, v &=& F(v)\,\,\,\,\,\,\,\,\,\,\,\,\,\,\, \,\,\,\,\,\,\,\,\,\,\,\,\,\,\,\,\,\,\,\,\,\,\,\,\,\,\,\textnormal{on}\,\, \mathbb{R}^{4+1}\\
v(0, x) = v_0 & \textnormal{and}& \ptl_T v(0, x)=v_1 \,\,\,\,\,\,\,\,\,\,\,\,\,\,\,\,\,\textnormal{on}\,\, \mathbb{R}^4\\\end{array} 
\right\}
\end{equation}
where $v$ is coupled to the 2+1 Einstein's equations with $u = Rv.$

Following our previous definition, the equation \eqref{n+1example} scatters in time if there exists a $v_S$ such that 
\begin{equation}
\left. \begin{array}{rcl}
\leftexp{4+1} {\square} \, v_S &=& 0\,\,\,\,\,\,\,\,\,\,\,\,\,\,\,\,\,\,\,\,\,\,\,\,\,\, \,\,\,\,\,\,\,\,\,\,\,\,\,\,\,\,\,\,\,\,\,\,\,\,\,\textnormal{on}\,\, \mathbb{R}^{4+1}\\
v_S (0, x) = (v_S)_0 & \textnormal{and}& \ptl_T v (0, x)=(v_S)_1 \,\,\,\,\,\,\,\,\,\textnormal{on}\,\, \mathbb{R}^4\\\end{array} 
\right\}
\end{equation}

\begin{align}
\Vert v - v_S \Vert_{E_v} \to 0 \quad \text{as}\quad T \to \infty
\end{align}

where $(v_S)_0$ and $(v_S)_1$ are the scattering data
\begin{align}
(v_S)_0 =& \mathcal{F}^{-1} \left(\hat{v}_0 - \int^\infty_0 \frac{\sin \sqrt{\Delta} s}{\sqrt{\Delta}} \hat{F}(s)ds \right) \\
(v_S)_1=& \mathcal{F}^{-1} \left( \hat{v}_1 - \int^\infty_0 \cos (\sqrt{\Delta} s) \hat{F} (s) ds \right)
\end{align}
and the energy topology of the problem \eqref{n+1example} is $\dot{H} (\mathbb{R}^4) \times L^2 (\mathbb{R}^4).$ We shall use the scattering theory for $v$ to study the asymptotic properties of $u.$


\begin{lemma}\label{Z-rep}
	Suppose $Z$ is a solution of the wave equation, 
	\begin{align}
	\leftexp{1+1}{\square} Z = \frac{f^2(u)}{r^2} - \ptl_\xi u \ptl_\eta u
	\end{align}
	then $Z$ satisfies:
	\begin{align}\label{2.15}
	Z(T, R) =&  c_{1} \int_{0}^{R} v^{2}(s, T + R - s) s ds + c_{2} \int_{0}^{R} v^{2}(s, T - R + s) s ds \notag\\ 
	&+ c_{3} R^{2} v^{2}(R, T) + c_{4} \int_{0}^{R}  \int_{T - (R + s)}^{T + (R - s)} v^{2} ds dT'  \notag\\
 &+ c_{5}  \int_{0}^{R}  \int_{T - (R + s)}^{T + (R - s)} F(v) v s^{2} ds dT'.
	\end{align}
	for some constants $c_i, i = 1 \cdots 5,$ assuming that $\displaystyle v=\frac{u}{R}$ satisfies \eqref{n+1example}. 
\end{lemma}
\begin{proof}
	\noindent Now since $u = R v$, by the product rule
	
	\begin{equation}\label{2.3}
	-\frac{1}{2} (\partial_{\eta} u)(\partial_{\xi} u) = -\frac{1}{2} R^{2} (\partial_{\eta} v)(\partial_{\xi} v) - \frac{1}{2} R (\partial_{\eta} v) v - \frac{1}{2} R (\partial_{\xi} v) v - \frac{1}{2} v^{2}.
	\end{equation}
	
	\noindent Making a change of variables,
	
	\begin{align}\label{2.4}
	\int_{0}^{R} \int_{T - (R - s)}^{T + (R - s)} s (\partial_{\eta} v) v dt ds =& \frac{1}{2} \int_{T - R}^{T + R} \int_{T - R}^{\xi} (\xi - \eta) (\partial_{\eta} v) v d\eta d\xi 
	\notag\\
	&= \frac{1}{4} \int_{T - R}^{T + R} \int_{T - R}^{\xi} (\xi - \eta) \partial_{\eta}(v^{2}) d\eta d\xi.
	\end{align}
	
	\noindent Integrating by parts,
	
	\begin{equation}\label{2.5}
	= \frac{1}{4} \int_{T - R}^{T + R} (\xi - (T - R)) v^{2}(T - R, \xi) d\xi - \frac{1}{4} \int_{T - R}^{T + R} \int_{T - R}^{\xi} v^{2}(\eta, \xi) d\eta d\xi.
	\end{equation}
	
	\noindent By a similar calculation,
	
	\begin{align}
	\int_{0}^{R} \int_{T - (R - s)}^{T + (R - s)} s (\partial_{\xi} v) v dt ds =& \frac{1}{2} \int_{T - R}^{T + R} \int_{\eta}^{T + R} (\xi - \eta) (\partial_{\xi} v) v d\xi d\eta 
	\notag\\
	=& \frac{1}{4} \int_{T - R}^{T + R} \int_{\eta}^{T + R} (\xi - \eta) \partial_{\xi}(v^{2}) d\xi d\eta \label{2.7}\\
	=& \frac{1}{4} \int_{T - R}^{T + R} (T + R - \eta) v^{2}(\eta, T + R) d\eta \notag \\
	&- \frac{1}{4} \int_{T - R}^{T + R} \int_{\eta}^{T + R} v^{2}(\eta, \xi) d\xi d\eta.\label{2.6}
	\end{align}
	
	\noindent Next, integrating by parts,
	
	\begin{equation}\label{2.8.0}
	\frac{1}{2} \int_{T - R}^{T + R} \int_{\eta}^{T + R} (\xi - \eta)^{2} (\partial_{\xi} v)(\partial_{\eta} v) d\xi d\eta = \frac{1}{2} \int_{T - R}^{T + R} (T + R - \eta)^{2} v (\partial_{\eta} v)(T + R, \eta) d\eta
	\end{equation}
	
	\begin{equation}\label{2.9.0}
	-\frac{1}{2} \int_{T - R}^{T + R} \int_{\eta}^{T + R} (\xi - \eta)^{2} v (\partial^2_{\xi \eta} v) d\xi d\eta - \int_{T - R}^{T + R} \int_{\eta}^{T + R} (\xi - \eta) v (\partial_{\eta} v) d\xi d\eta.
	\end{equation}
	
	\noindent Now by a change of variables,
	
	\begin{equation}\label{2.10}
	\frac{1}{4} \int_{T - R}^{T + R} \int_{\eta}^{T + R} (\xi - \eta) v (\partial_{\eta} v) d\xi d\eta = (\ref{2.4}).
	\end{equation}
	
	\noindent Also, integrating by parts,
	
	\begin{align}\label{2.11}
	&\frac{1}{2} \int_{T - R}^{T + R} (T + R - \eta)^{2} v (\partial_{\eta} v)(T + R, \eta) d\eta \notag\\ 
	&= -\frac{1}{4} (2R)^{2} v^{2}(T + R, T - R) + \frac{1}{2} \int_{T - R}^{T + R} (T + R - \eta) v^{2}(T + R, \eta) d\eta.
	\end{align}
	
	\noindent Finally, 
	
	\begin{equation}\label{2.12}
	\aligned
	-\frac{1}{2} \int_{T - R}^{T + R} \int_{\eta}^{T + R} (\xi - \eta)^{2} v (\partial^2_{\xi \eta} v) d\xi d\eta = -\frac{3}{4} \int_{T - R}^{T + R} \int_{\eta}^{T + R} (\xi - \eta) v(\partial_{R} v) d\xi d\eta \\ + \frac{1}{8} \int_{T - R}^{T + R} \int_{\eta}^{T + R} (\xi - \eta)^{2} v F(v) d\xi d\eta,
	\endaligned
	\end{equation}
where we used the fact that $\displaystyle \leftexp{4+1}{\square} v = -4 \ptl^2_ {\xi \eta} v + \frac{3}{R} \ptl_R V = F(v).$ In particular, note that 
\begin{align}
-\halb \ptl^2_{\xi \eta} v = - \frac{3}{4 (\xi -\eta)} \ptl_R v + \frac{1}{8} F(v).  
\end{align}
	
	\noindent Now, making a change of variables,
	
	\begin{equation}\label{2.13}
	-\frac{3}{4} \int_{T - R}^{T + R} \int_{\eta}^{T + R} (\xi - \eta) v (\partial_{R} v) d\xi d\eta = -\frac{3}{4} \int_{0}^{R} \int_{T - (R + s)}^{T + (R - s)} v (\partial_{s} v) s ds dT'.
	\end{equation}
	
	\noindent Integrating by parts,
	
	\begin{align}\label{2.14.0}
	&= -\frac{3}{8} \int_{0}^{R} v^{2}(s, T + R - s) s ds -\frac{3}{8} \int_{T - R}^{T} v^{2}(s, T - R + s) s ds  \notag\\
	&+ \frac{3}{4} \int_{0}^{R} \int_{T - (R + s)}^{T + (R - s)} v^{2} ds dT'.
	\end{align}
	
	\noindent Therefore we have 
	
	\begin{equation}\label{2.15v2}
	\aligned
	Z =  c_{1} \int_{0}^{R} v^{2}(s, T + R - s) s ds + c_{2} \int_{0}^{R} v^{2}(s, T - R + s) s ds \\ + c_{3} R^{2} v^{2}(R, T) + c_{4} \int_{0}^{R}  \int_{T - (R + s)}^{T + (R - s)} v^{2} ds dT' + c_{5}  \int_{0}^{R}  \int_{T - (R + s)}^{T + (R - s)} F(v) v s^{2} ds dT'.
	\endaligned
	\end{equation}
\end{proof}

Consider the radial wave equation, 
\begin{align}\label{4+1gen}
\leftexp{4+1}{\square} v = F (v), \quad \text{on} \quad  (\mathbb{R}^{4+1}, \check{g})
\end{align}
and the energy-momentum tensor that arises directly from the variational principle of \eqref{4+1gen}
\begin{align}
\check{T}_{\mu \nu} \fdg= \grad_\mu v \grad_\nu v - \halb \check{g}_{\mu \nu} \grad^\sigma v \grad_\sigma v - \check{g}_{\mu \nu} \tilde{F} (v),\quad \mu, \nu, \sigma = 0,1 \cdots 4.
\end{align}
where $\tilde{F}(v)$ is such that its variational derivatice with respect to $v$ is  $F(v)$; in other words, $\tilde{F} (v)$ is such that the variational principle corresponding to the equation \eqref{4+1gen} is 

\begin{align}
    \int \left(  -\halb \check{g}^{ \mu \nu} \ptl_\mu v \ptl_\nu v - \tilde{F} (v)  \right) \bar{\mu}_{\check{g}}. 
\end{align}

The divergence of $\check{T}$ is given by 
\begin{align}
\grad^\nu \check{T}_{\mu \nu} =& (\ptl^\nu \ptl_\mu v) \ptl_\nu v + \ptl_\mu v \leftexp{4+1}{\square} v 
- \halb \check{g}_{\mu \nu} \ptl^\nu (\ptl^\sigma v \ptl_\sigma v) - \check{g}_{\mu \nu} \ptl^\nu \tilde{F}(v) \notag\\
=& \ptl_\mu v (\leftexp{4+1}{\square} v) - \check{g}_{\mu \nu} \ptl^\nu \tilde{F}(v)
\end{align}

\noindent Consider the vector field $\mathfrak{X}$ such that its momentum is given by 
\begin{align}
J_{\mathfrak{X}} = \check{T} (\mathfrak{X}) \quad  \textnormal{i.e.,}  \quad (J_{\mathfrak{X}} )^\nu = \check{T}^\nu_\mu \mathfrak{X}^\mu,
\end{align}
so that, after relabeling the indices for convenience, we have the identity
\begin{align}
\grad_\nu J_{\mathfrak{X}}^\nu = \halb \leftexp{\mathfrak{(X)}}{\pi}_{\mu \nu} \check{T}^{\mu \nu} + \mathfrak{X}^\mu\grad^\nu \check{T}_{\mu\nu},
\end{align}
where $\leftexp{\mathfrak{(X)}}{\pi}_{\mu \nu}$ is the deformation tensor defined as 

\begin{align}
 \leftexp{(\mathfrak{X})}{\pi}_{\mu \nu} \fdg = \check{g}_{\sigma \nu} \ptl_\mu \mathfrak{X}^{\sigma} +
\check{g}_{\sigma \mu} \ptl_\nu \mathfrak{X}^{\sigma} + \mathfrak{X}^{\sigma} \ptl_\sigma \check{g}_{\mu \nu}. 
\end{align}
Define $\check{e}$ and $\check{m}$ such that 
\begin{align}
\check{e} + \tilde{F}(v) = \check{T}(\ptl_T, \ptl_T),\quad \check{m} \fdg = \check{T}(\ptl_T, \ptl_R).
\end{align}
Now consider a Morawetz multiplier vector $\mathfrak{X} \fdg = \mathfrak{F}(R) \ptl_R$ so that the corresponding momentum
is given by 

\begin{align}
J_\mathfrak{X} =&\, \check{\mathbf{T}}(\mathfrak{X}) \notag \\
=&\, \mathfrak{F}(R) \big( -\check{m} \ptl_T + (\check{e} + \tilde{F}(v)) \ptl_R \big) 
\end{align}

\noindent and its divergence

\begin{align}\label{mora1_div}
\grad_\nu J^\nu_\mathfrak{X} = \halb \check{\mathbf{T}}^{\mu \nu }\,\leftexp{(\mathfrak{X})}{\pi_{\mu \nu}}
+ \mathfrak{X}^\mu \ptl_\mu v (\leftexp{4+1}{\square} v) - \mathfrak{X}^\mu\check{g}_{\mu \nu} \ptl^\nu \tilde{F}(v),
\end{align}
where the non-zero terms of deformation tensor 
\[\leftexp{(\mathfrak{X})}{\pi}_{\mu \nu} \fdg = \check{g}_{\sigma \nu} \ptl_\mu \mathfrak{X}^{\sigma} +
\check{g}_{\sigma \mu} \ptl_\nu \mathfrak{X}^{\sigma} + \mathfrak{X}^{\sigma} \ptl_\sigma \check{g}_{\mu \nu}\]
are given by
\begin{align*}
\leftexp{(\mathfrak{X})}{\pi_{RR}} = 2 g_{RR} \ptl_R \mathfrak{F}(R),&  \leftexp{(\mathfrak{X})}{\pi_{\theta \theta}} =  \frac{2}{R} g_{\theta \theta} \mathfrak{F} (R), \notag \\  
\leftexp{(\mathfrak{X})}{\pi_{\phi\phi} }=  \frac{2}{R} g_{\phi \phi} \mathfrak{F}(R),& \leftexp{(\mathfrak{X})}{\pi_{\psi \psi}} =  \frac{2}{R} g_{\psi \psi} \mathfrak{F}(R).
\end{align*}
Consequently a calculation shows that  \eqref{mora1_div} can be represented as
\begin{align} \label{moragen}
\grad_\nu J^\nu_\mathfrak{X} = \left(  - \frac{6\mathfrak{F}(R)}{R}\check{\cal{L}} + \check{e}\, \ptl_R \mathfrak{F} (R) \right) + \mathfrak{F}(R)\ptl_R v F(v) - \mathfrak{F}(R) \ptl^R \tilde{F}(v)
\end{align}
Now define the following lower-order momentum vector
\begin{align}
J^\nu_1 [v] \fdg = \kappa   v \grad^\nu v - \halb
v^2 \grad ^\nu \kappa + \tilde{F}(v) \mathfrak{X}^\nu .
\end{align}
Its divergence is 
\begin{align}
\grad_\nu J^\nu_1 = &\kappa v \leftexp{4+1}{\square} v + \kappa \grad^\nu v \grad_\nu v  + v \grad^\nu v \grad_\nu \kappa - (\square \kappa )\frac{v^2}{2} \notag\\
&\quad- v \grad^\nu \kappa  \grad_\nu v + \mathfrak{X}^\nu \ptl_\nu \tilde{F}(v) +
\tilde{F} (v)\grad_\nu \mathfrak{X}^\nu \notag \\
=&\kappa v \leftexp{4+1}{\square} v+\kappa \grad^\nu v \grad_\nu v   - (\square \kappa )\frac{v^2}{2} + \mathfrak{X}^\nu \ptl_\nu \tilde{F}(v) +
\tilde{F}(v) \grad_\nu \mathfrak{X}^\nu\notag \\
=&\kappa v \leftexp{4+1}{\square} v+ 2 \kappa \, \check{\cal{L}} - (\square \kappa )\frac{v^2}{2}+ \mathfrak{X}^\nu \ptl_\nu \tilde{F}(v) +
\tilde{F}(v) \grad_\nu \mathfrak{X}^\nu.
\end{align}
In addition, consider a timelike multiplier vector field $\mathfrak{T} =\ptl_T$ such that the corresponding current is 
\begin{align}
J_T \fdg =& \check{T}(\mathfrak{T}) \notag\\
=& (-\check{e} -\tilde{F}(v)) \ptl_T + \check{m} \ptl_R 
\end{align}
the divergence 
\begin{align}
\grad_\nu J_T^\nu =& \halb \leftexp{(\mathfrak{T})}{\pi}_{\mu \nu} \check{T}^{\mu \nu} + (\mathfrak{T})^\mu \grad^\nu \check{T}_{\mu \nu} \notag\\
=& \ptl_T v F(v) + \ptl^T \tilde{F}(v)
\end{align}
in view of $\leftexp{(\mathfrak{T})}{\pi}_{\mu \nu} \equiv 0.$ Likewise, if we define a lower-order momentum vector field 
\begin{align}
J_2^\nu \fdg=& \tilde{F}(v) \mathfrak{T}^\nu \notag\\
\intertext{then}
\grad_\nu J_2^\nu =& \mathfrak{T}^\nu \ptl_\nu \tilde{F}(v) + \tilde{F}(v) \grad_\nu \mathfrak{T}^\nu.
\end{align}

\noindent Consider the volume $3-$form of $(M, g):$
\begin{align}
\bar{\mu}_g = \halb r e^{2Z} d\eta \wedge d\xi \wedge d \theta
\end{align}
Define the $2-$ forms $\bar{\mu}_\xi$ and $\bar{\mu}_\eta$ such that 
\begin{align}
d \xi \wedge \bar{\mu}_\xi \fdg =& \bar{\mu}_g \\
d \eta \wedge \bar{\mu}_\eta \fdg =& \bar{\mu}_g
\end{align}
so that we have
\begin{align}
\bar{\mu}_\xi =& -\halb r e^{2Z} (d\eta \wedge d \theta)\\
\bar{\mu}_\eta=& \halb r e^{2Z} (d \xi \wedge d\theta)
\end{align}
explicitly. Suppose $ Y $ is smooth vector field defined on $(M, g)$ then the flux through the $\eta =c$ and $\xi=c$ null hypersurfaces are
\begin{align}
\text{Flux}^+(Y) \fdg = \int_{\{ \eta =c\}} d\eta (Y)\, \bar{\mu}_\eta\\
\text{Flux}^-(Y) \fdg = \int_{\{\xi =c\}} d \xi (Y)\, \bar{\mu}_\xi
\end{align}
respectively. Now having fixed the orientation, let us now give the equivalent definitions for the corresponding 4+1 Minkowski metric $(\mathbb{R}^{4+1}, \check{g})$ that are consistent with the original spacetime $(M, g):$
\begin{align}
\check{\mu}_{\check{g}} = \halb \sqrt{-\check{g}}\, d\eta \wedge d\xi \wedge d\omega_{\mathbb{S}^3}.
\end{align}
Define $\check{\mu}_\eta$ and $\check{\mu}_\xi$ such that, 
\begin{align}
d \xi \wedge \check{\mu}_\xi \fdg =& \check{\mu}_{\check{g}} \\
d \eta \wedge \check{\mu}_\eta \fdg =& \check{\mu}_{\check{g}},
\end{align}
and for a smooth vector field $P$ in $(\mathbb{R}^{4+1}, \check{g})$ the fluxes through $\eta=c$ and $\xi=c$ hypersurfaces are
\begin{align}
\text{Flux}^+(P) \fdg = \int_{\{ \eta =c\}} d\eta (P)\, \bar{\mu}_\eta\\
\text{Flux}^-(P) \fdg = \int_{\{\xi =c\}} d \xi (P)\, \bar{\mu}_\xi.
\end{align}
If we consider the vector field $J_T,$ we have
\begin{subequations}
	\begin{align}
	\text{Flux}^+ (J_T)=& \int_{\{ \eta=c \}} -(\check{e} + \tilde{F}(v) +\check{m}) \check{\mu}_\eta  \\
	\text{Flux}^-(J_T)=&\int_{\{\xi =c\}} -(\check{e} +\tilde{F}(v) - \check{m})\check{\mu}_\xi \\
	\text{Flux}^+ (J_2)=& \int_{\{ \eta=c \}} ( \tilde{F}(v)) \check{\mu}_\eta  \\
	\text{Flux}^-(J_2)=&\int_{\{\xi =c\}} -(\tilde{F}(v))\check{\mu}_\xi 
	\end{align}
\end{subequations}

\begin{lemma}
	Suppose $(M, g, U)$ is a globally regular solution of the initial value problem of the Einstein-equivariant wave map system with $E_0 < \eps^2,$ for $\eps$ sufficiently small, we have
	\begin{subequations}
		\begin{align}
		\text{Flux}^+(J) \geq&\, 0, \\
		\text{Flux}^-(J) \leq&\, 0, 
		\end{align}
	\end{subequations}
	where  $J \fdg = J_T + J_2.$
\end{lemma}
\begin{theorem}[Nonlinear Morawetz Estimate]\label{t1.1}
	Suppose $v$ is a globally regular solution of the equation 
	\begin{equation}
	\left. \begin{array}{rcl}
	\leftexp{4+1}{\square}\, v &=& F(v)\,\,\,\,\,\,\,\,\,\,\,\,\,\,\, \,\,\,\,\,\,\,\,\,\,\,\,\,\,\,\,\,\,\textnormal{on}\,\, \mathbb{R}^{4+1}\\
	v_0 = v (0, x) & \textnormal{and}& v_1 = \ptl_T v (0, x) \,\,\,\,\,\,\,\, \textnormal{on}\,\, \mathbb{R}^4\\\end{array} 
	\right\}
	\end{equation}
	where 
	\begin{align}
	F(v) \fdg=& \left(e^{2Z} - 1 + \left(\frac{r}{R}\partial_{\eta} r + \frac{1}{2}\right) - \left(\frac{r}{R}\partial_{\xi} r - \frac{1}{2}\right)\right) \frac{v}{r^{2}}  + 2 \ptl_\xi v \ptl_\eta \log \left(\frac{r}{R}\right) \notag\\ &+ 2 \ptl_\eta v \ptl_\xi \log\left(\frac{r}{R}\right)  + e^{2Z} \frac{R^{2}}{r^{2}} v^{3} \zeta(R v)
	\end{align}
	with $Z$ and $r$ are coupled as

 \begin{align}
     -4 \ptl_\xi \ptl_\eta Z =& \frac{1}{2} \left(4 \ptl_\eta \ptl_\eta u + e^{2Z} \frac{f^2 (u)}{r^2} \right) \\
r^{-1} \ptl_\eta \ptl_ \eta r =& \frac{e^{2Z}}{4} \frac{f^2(u)}{r^2},
 \end{align}
 which are basically the $\mathbf{E}_{\theta \theta} =\mathbf{T}_{\theta \theta}$ and $\mathbf{E}_{\xi \eta} =\mathbf{T}_{\xi \eta}$ Einstein equations. Recall that $u= Rv$;
 then 
	\begin{align} \label{morawetz-estimate}
	\int_{\mathbb{R}^{4+1}} \frac{v^2}{\vert x \vert^3} \bar{\mu}_{\check{g}} \leq \Vert v_0 \Vert^2_{\dot{H}^1 (\mathbb{R}^4)} + \Vert v_1 \Vert^2_{L^2(\mathbb{R}^4)}.
	\end{align}
\end{theorem}

\begin{proof}
	To prove the result, we shall use the Morawetz multiplier method we introduced above. Recall that $\tilde{F} (v)$ is such that its first variational derivative is $F(v)$ i.e.,  the variational principle of 

 \begin{align}
     \leftexp{4+1}{\square} v = F (v) 
 \end{align}
 is 
 \begin{align}
     \int \left(  -\halb \check{g}^{ \mu \nu} \ptl_\mu v \ptl_\nu v - \tilde{F} (v)  \right) \bar{\mu}_{\check{g}}. 
\end{align}

 Consider the multiplier $\mathfrak{X}$ such that, $\mathfrak{F}(R) = \frac{1}{3}$, we have
	
	\begin{align}
	J_{\mathfrak{X}} =& \frac{1}{3} ( -\check{m} \ptl_T + (\check{e} - \tilde{F}(v)) \ptl_R)
	\intertext{and}
	\grad_\nu J^\nu_{\mathfrak{X}} =& -\frac{2}{R} \check{\mathcal{L}} + \frac{1}{3} \ptl_R v F(v) - \mathfrak{X}^\mu \check{g}_{\mu \nu} \ptl^\nu \tilde{F}(v).
	\end{align}
	Likewise, if we consider the choice of $\kappa= \frac{1}{R}$, then 
	\begin{align}
	\leftexp{4+1}{\square} \kappa =& -\frac{1}{R^3} \\
	J_\kappa =& - \left(\frac{1}{R} v \ptl_T v \right) \ptl_T + \left(\frac{1}{R} v \ptl_R v + \frac{v}{R^2} + \mathfrak{X}^R \tilde{F}(v) \right) \ptl_R \\
	\grad_\nu J^\nu_\kappa =& \frac{2}{R} \check{\mathcal{L}}+ \frac{1}{R} v F(v)+ \frac{1}{2} \frac{v^2}{R^3} + \mathfrak{X}^\nu \ptl_\nu \tilde{F}(v)
	\end{align}
	Subsequently, if we consider the sum vector $J^\nu_S \fdg= J^\nu_{\mathfrak{X}} + J^\nu_\kappa$, then we have
	\begin{align}
	\grad_\nu J^\nu_S =& \grad_\nu J_{\mathfrak{X}}^\nu +\grad_\nu J_{\kappa}^\nu \notag\\
	=&\frac{1}{2} \frac{v^2}{R^3} + \frac{1}{3} \ptl_R v F(v) + \frac{1}{R} v F(v)
	\end{align}
	We have the following identity from the Stokes theorem, for the region enclosed between 
	$\check{\Sigma}_0$ and $\check{\Sigma}_T$,
	
	\begin{align}
	&\halb \int \frac{v^2}{R^3}\, \bar{\mu}_{\check{g}} +  \frac{1}{3} \int \ptl_R v F(v) \bar{\mu}_{\check{g}} + \int \frac{1}{R} v F(v) \bar{\mu}_{\check{g}}  \notag \\
 = & \int  \grad_\nu J^\nu_S \bar{\mu}_{\check{g}} = \int_{ {\check{\Sigma}}_0}  \ip{\ptl_T}{J_S}\bar{\mu}_{\check{q}} -\int_{\check{\Sigma}_T} \ip{\ptl_T}{ J_S} \bar{\mu}_{\check{q}}.
	\end{align}
	
	\noindent If we consider the boundary terms on  $\check{\Sigma}_0$ and $\check{\Sigma}_T$
	\begin{align}
	\check{g}_{\mu \nu} J^\mu_S (\ptl_T)^\nu = \check{g}_{\mu \nu} J^\mu_{\mathfrak{X}} (\ptl_T)^\nu + \check{g}_{\mu \nu} J^\mu_{\kappa} (\ptl_T)^\nu  
	\end{align}   
	where
	
	\begin{align}
	\check{g}_{\mu \nu} J^\mu_{\mathfrak{X}} (\ptl_T)^\nu = \check{g}_{T T} J_{\mathfrak{X}}^T (\ptl_T)^T = \mathfrak{F}(R) \check{m} 
	\intertext{and}
	\check{g}_{\mu \nu} J^\mu_{\kappa} (\ptl_T)^\nu = \check{g}_{T T} J_{\kappa}^T (\ptl_T)^T = -\frac{1}{\vert  x \vert} v \, \ptl_T v
	\end{align}
 it follows from an algebraic manipulation that 
	
	\begin{align}
	\int_{ {\check{\Sigma}}_0}  \ip{\ptl_T}{J_S}\bar{\mu}_{\check{q}},\quad  \int_{ {\check{\Sigma}}_T}  \ip{\ptl_T}{J_S}\bar{\mu}_{\check{q}}
	\leq \Vert v_0 \Vert^2_{\dot{H}^1 (\mathbb{R}^4) } + 
	\Vert v_1 \Vert^2_{L^2 (\mathbb{R}^4)}.
	\end{align}
	
	Thus, 

	\begin{equation}\label{1.8}
	\int_{\mathbb{R}^{4+1}} \frac{v^{2}}{R^3} \bar{\mu}_{\check{g}}\leq E(v) + \Big\vert\int \int F(v) v R^{2} dR dT
	\Big\vert + \Big\vert \int \int F(v) v_{R} R^{3} dR dT\Big\vert.
	\end{equation}
	
	\noindent Therefore, to prove the theorem it suffices to absorb the last two terms in the right hand side of $(\ref{1.8})$ into the left hand side of $(\ref{1.8})$. The term
	
	\begin{equation}\label{1.17}
	\int_{\mathbb{R}^{4+1}} \frac{F(v) v}{ \vert x \vert} \bar{\mu}_{\check{g}}
	\end{equation}
	
	\noindent is easier. Making a change of variables, 
	
	\begin{equation}
	\int_{\mathbb{R}} \int_{\mathbb{R}} F(v) v R^{2} dR dT = \frac{1}{4} \int_{\mathbb{R}} \int_{\mathbb{R}} F(v(\xi, \eta)) v(\xi, \eta) (\xi - \eta)^{2} d\xi d\eta.
	\end{equation}
	\noindent Now since (see for example, Corollary 5.11 in \cite{AGS_15})
	
 \begin{align} \Big\vert \frac{R}{r}  -1 \Big\vert , \vert \partial_{\eta} r - \frac{1}{2} \vert, \vert \partial_{\xi} r + \frac{1}{2}\vert \leq \epsilon(E(v)),
	\end{align}
 
 where $\epsilon(E(v)) \searrow 0$ as $E \searrow 0$, the Cauchy-Schwartz inequality 
	
	\begin{align}\label{1.18}
	&\int_{\mathbb{R}} \int_{\mathbb{R}} \Big\vert\partial_{\eta} r - \frac{1}{2}\Big\vert \Big \vert\partial_{\xi} v(\xi, \eta)\Big \vert \vert v(\xi, \eta) \vert \vert \xi - \eta \vert d\xi d\eta \notag\\
	&\leq \int_{\mathbb{R}} \int_{\mathbb{R}}  \Big \vert\partial_{\eta} r - \frac{1}{2}\Big \vert \Big \vert\partial_{\xi} v(\xi, \eta)\Big\vert^{2} \vert \xi - \eta \vert^{2} d\xi d\eta + \int_{\mathbb{R}} \int_{\mathbb{R}} \Big \vert\partial_{\eta} r - \frac{1}{2} \Big\vert  \vert v(\xi, \eta) \vert ^{2} d\xi d\eta.
	\end{align}
	
	\noindent Now again by a change of variables
	
	\begin{equation}\label{1.18.1}
	\int_{\mathbb{R}} \int_{\mathbb{R}} \Big \vert\partial_{\eta} r - \frac{1}{2}\Big \vert  \vert v(\xi, \eta) \vert^{2} d\xi d\eta \leq \epsilon(E(v)) \int_{\mathbb{R}} \int_{\mathbb{R}} \vert v(R, T) \vert^{2} dR dT.
	\end{equation}
	
	\noindent It is also straightforward to show that
	
	\begin{align}\label{1.18.2}
	&\int_{\mathbb{R}} \int_{\mathbb{R}} \Big \vert \partial_{\eta} r - \frac{1}{2}\Big \vert \vert \partial_{\xi} v(\xi, \eta) \vert ^{2} d\xi d\eta \notag\\
	&\leq \left(\int_{\mathbb{R}} \left(\sup_{\xi} \frac{1}{\vert \xi - \eta\vert} \vert \partial_{\eta} r - \frac{1}{2} \vert \right) d\eta \right) \cdot \sup_{\eta} \left(\int \vert \xi - \eta \vert^{3} \vert \partial_{\xi} v(\xi, \eta)\vert^{2} d\xi \right). 
	\end{align}
	\noindent By energy bounds on the boundary of a light cone, since $ \halb \vert \xi - \eta \vert = R$,
	
	\begin{equation}\label{1.18.3}
	\int_{\mathbb{R}} \vert \xi - \eta \vert^{3} \vert \partial_{\xi} v(\xi, \eta)\vert^{2} d\xi \leq E(v).
	\end{equation}
\noindent Note that $\vert \ptl_\xi \ptl_\eta r \vert \leq R \vert v \vert^2.$ This estimate is based on the Einstein equation $$\displaystyle  \ptl_\xi \ptl_\eta r = r \left( \frac{e^{2Z}}{4} \frac{f^2 (u)}{r^2} \right)$$ 
and the result $r \leq R, \vert Z \vert, \vert u \vert  \leq c (E_0);$ (see for example Corollary 5.10 and Corollary 5.11 in \cite{AGS_15}) and the fact that $f$ is an odd function with $f(0) =0$ and $f_u (u) =1.$ Also recall that $u=Rv.$

	\noindent Now, by the fundamental theorem of calculus, the fact that $\partial_{\eta} r\Big\vert_{R = 0} = \frac{1}{2}$, and $\vert \partial_{\xi} \partial_{\eta} r \vert \leq R \vert v \vert^{2}$,
	
	\begin{equation}\label{1.18.4}
	\sup_{\xi} \frac{1}{\vert \xi - \eta \vert} \Big \vert(\partial_{\eta} r - \frac{1}{2}) \Big \vert \leq \int_{\mathbb{R}} \vert v(\xi, \eta)\vert ^{2} d\xi.
	\end{equation}
	
	\noindent Therefore,
	
	\begin{equation}\label{1.18.5}
	(\ref{1.18}) \leq E(v) \int_{\mathbb{R}} \int_{\mathbb{R}} \vert v(\xi, \eta) \vert ^{2} d\xi d\eta.
	\end{equation}

	\noindent By a similar calculation,
	
	\begin{align}\label{1.19}
	&\int_{\mathbb{R}} \int_{\mathbb{R}} \frac{1}{R} \Big \vert\partial_{\xi} r + \frac{1}{2}\Big \vert \Big\vert\partial_{\eta} v(R, T)\Big \vert \vert v(R, T)\vert R^{2} dR dT \notag\\
	&\leq E(v) \int_{\mathbb{R}} \int_{\mathbb{R}} v(R, T)^{2} dR dT + \epsilon(E(v)) \int_{\mathbb{R}} \int_{\mathbb{R}} v(R, T)^{2} dR dT.
	\end{align}
	
	\noindent Also, since $\vert Z \vert \leq \epsilon(E(v))$,
	
	\begin{align}\label{1.20}
	&\int_{\mathbb{R}} \int_{\mathbb{R}} \left(e^{2Z} - 1 + \partial_{\eta} r + \frac{1}{2} - (\partial_{\xi} r - \frac{1}{2})\right) \frac{v(R, T)^{2}}{R^{2}} R^{2} dR dT  \notag\\ 
	&\leq \epsilon(E(v)) \int_{\mathbb{R}} \int_{\mathbb{R}} v(R, T)^{2} dR dT.
	\end{align}
	
	\noindent Therefore, combining $(\ref{1.18.1})$, $(\ref{1.18.5})$, and $(\ref{1.20})$,
	
	\begin{equation}\label{1.21}
	(\ref{1.17}) \leq E(v) \int_{\mathbb{R}} \int_{\mathbb{R}} v(R, T)^{2} dR dT + \epsilon(E(v)) \int_{\mathbb{R}} \int_{\mathbb{R}} v(R, T)^{2} dR dT.
	\end{equation}

	\noindent The bulk term
	
	\begin{equation}\label{1.10}
	\int_{\mathbb{R}^{4+1}}  F(v) \ptl_R v \, \bar{\mu}_{\check{g}}
	\end{equation}
	
	\noindent is more difficult to control. Several components of this term are essentially the same as in $(\ref{1.17})$. Indeed, as in $(\ref{1.18.3})$ - $(\ref{1.18.5})$,
	
	\begin{equation}\label{1.11}
	\aligned
	\int_{\mathbb{R}^{4+1}} \frac{1}{R} \Big \vert \partial_{\eta} r - \frac{1}{2}\Big \vert \Big \vert\partial_{\xi} v(R, T)\Big \vert \Big\vert\partial_{R} v(R, T)\Big \vert \bar{\mu}_{\check{g}} \leq E(v) \int_{\mathbb{R}} \int_{\mathbb{R}} v(R, T)^{2} dR dT,
	\endaligned
	\end{equation}
	
	\noindent and
	
	\begin{equation}\label{1.12}
	\int_{\mathbb{R}^{4+1}} \frac{1}{R} \Big\vert \partial_{\xi} r + \frac{1}{2}\Big\vert \Big\vert\partial_{\eta} v(R, T)\Big\vert \Big\vert\partial_{R} v(R, T)\Big\vert  \bar{\mu}_{\check{g}} \leq E(v) \int_{\mathbb{R}} \int_{\mathbb{R}} v(R, T)^{2} dR dT.
	\end{equation}
	
	\noindent Also,
	
	\begin{align}\label{1.13}
	&\int_{\mathbb{R}^{4+1}} \Big\vert\partial_{\eta} r + \frac{1}{2}\Big\vert \Big\vert\frac{v(R, T)}{R^{2}}\Big\vert \vert v_{R}(R, T)\vert  \bar{\mu}_{\check{g}} \notag\\ 
	&\leq  \int_{\mathbb{R}^{4+1}} \frac{1}{R} \Big\vert\partial_{\eta} r + \frac{1}{2}\Big\vert v_{R}(R, T)^{2} \bar{\mu}_{\check{g}} + \int_{\mathbb{R}}\int_{\mathbb{R}} \Big\vert\partial_{\eta} r + \frac{1}{2}\Big\vert v(R, T)^{2} dR dT \notag\\
	&\leq E(v) \int_{\mathbb{R}} \int_{\mathbb{R}} v(R, T)^{2} dR dT + \epsilon(E(v)) \int_{\mathbb{R}} \int_{\mathbb{R}} v(R, T)^{2} dR dT,
	\end{align}
	
	\noindent and
	
	\begin{align}\label{1.14}
	&\int_{\mathbb{R}^{4+1}} \Big\vert\partial_{\xi} r - \frac{1}{2}\Big\vert \Big\vert \frac{v(R, T)}{R^{2}}\Big\vert \vert v_{R}(R, T)\vert \bar{\mu}_{\check{g}} \notag \\ 
	&\leq  E(v) \int_{\mathbb{R}} \int_{\mathbb{R}} v(R, T)^{2} dR dT + \epsilon(E(v)) \int_{\mathbb{R}} \int_{\mathbb{R}} v(R, T)^{2} dR dT.
	\end{align}
	
	\noindent It remains to compute
	
	\begin{equation}\label{1.15}
	\int_{\mathbb{R}^{4+1}} (e^{2Z} - 1) \frac{v(R, T)}{R^{2}} v_{R}(R, T) \bar{\mu}_{\check{g}}, 
	\end{equation}
	
	\noindent and
	
	\begin{equation}\label{1.15.1}
	\int_{\mathbb{R}^{4+1}} e^{2Z} \frac{R^{2}}{r^{2}} v(R, T)^{3} \zeta(Rv) v_{R}(R, T) \bar{\mu}_{\check{g}}.
	\end{equation}
	
	
Consider the Taylor's theorem for $e^{2Z}:$

\begin{align} \label{1.15.2}
    e^{2Z} = 1 + 2Z + 2 Z^{2} + \frac{4}{3} Z^{3} + \cdots + \frac{(2Z)^n e^{2 \chi}} {n !}
\end{align}
	\noindent for some $\chi$ between $0$ and $Z$. Since $ \vert Z \vert \leq \epsilon(E(v))$ it suffices to consider only $1 + 2Z$. The other terms follow in a similar manner. In view of Lemma \ref{Z-rep}, to estimate
	
	\begin{equation}\label{1.15.3}
	2 \int \int Z(R, T) \frac{v(R, T)}{R^{2}} v_{R}(R, T) R^{3} dR dT, 
	\end{equation}
	
	\noindent we shall split $Z(R, T) = Z_{1}(R, T) + Z_{2}(R, T)$, where $Z_{1}(R, T)$ has good integral properties and $Z_{2}(R, T)$ has a radial derivative with good integral properties. By equation $(7d)$,
	
	\begin{equation}\label{2.1}
	\partial_{\xi \eta}^{2} Z = -\frac{1}{2} (\partial_{\eta} u)(\partial_{\xi} u) - \frac{e^{2Z}}{8} \frac{f^{2}(u)}{r^{2}}.
	\end{equation}
	
	\noindent Then by Duhamel's principle and the fact that $Z = Z_{R} = 0$ on $R = 0$,
	
	\begin{equation}\label{2.2}
	Z(R, T) = \int_{0}^{R} \int_{T - (R - s)}^{T + (R - s)} (\partial_{\xi \eta}^{2} Z)(s, t) dt ds.
	\end{equation}

	\noindent Now the term
	
	\begin{equation}\label{2.1.1}
	-\int_{0}^{R} \int_{T - (R - s)}^{T + (R - s)} \frac{e^{2Z}}{8} \frac{f^{2}(sv)}{r^{2}} ds dt = Z_{2}^{(1)}(R, T)
	\end{equation}
	
	\noindent may be safely placed in $Z_{2}$. Indeed, integrating by parts,
	
	\begin{align}\label{2.25}
	&\int_{\mathbb{R}} \int_{\mathbb{R}} Z_{2}^{(1)}(R, T) v(R, T) v_{R}(R, T) R dR dT \\ 
	&= -\frac{1}{2} \int_{\mathbb{R}} \int_{\mathbb{R}} Z_{2}^{(1)}(R, T) v(R, T)^{2} dR dT - \frac{1}{2} \int_{\mathbb{R}} \int_{\mathbb{R}} \partial_{R} (Z_{2}^{(1)}(R, T)) v(R, T)^{2} R dR dT.
	\end{align}
	
	\noindent Since $ \vert Z_{2}^{(1)}(R, T) \vert \leq \int \int v(r, t)^{2} dr dt$, 
	
	\begin{equation}\label{2.26}
	\frac{1}{2} \int \int Z_{2}^{(1)} v(R, T)^{2} dR dT \leq \left(\int \int v(R',T')^{2} dR' dT' \right)^{2}.
	\end{equation}
	
	\noindent Meanwhile, we can directly compute
	
	\begin{equation}\label{2.27}
	\aligned
	-\partial_{R} (\int_{0}^{R} \int_{T - (R - s)}^{T + (R - s)} \frac{e^{2Z}}{8} \frac{f^{2}(sv)}{r^{2}} ds dT') = -\int_{0}^{R} \frac{e^{2Z}}{8} \frac{f^{2}(sv)}{r^{2}}(s, T + (R - s)) ds \\
	- \int_{0}^{R} \frac{e^{2Z}}{8} \frac{f^{2}(sv)}{r^{2}}(s, T - (R - s)) ds.
	\endaligned
	\end{equation}
	
	\noindent Since $ \vert Z \vert $ is small,
	
	\begin{equation}\label{2.27.1}
	\sup_{R} \int_{0}^{R} \frac{e^{2Z}}{8} \frac{f^{2}(sv)}{r^{2}}(s, T + (R - s)) ds \leq \int v^{2}(\eta, \xi = R + T) d\eta,
	\end{equation}
	
	\noindent and
	
	\begin{equation}\label{2.27.2}
	\sup_{R} \int_{0}^{R} \frac{e^{2Z}}{8} \frac{f^{2}(sv)}{r^{2}}(s, T - (R + s)) ds \leq \int v^{2}(\eta = T - R, \xi) d\xi.
	\end{equation}
	
	\noindent Then since Hardy's inequality implies
	
	\begin{equation}\label{2.27.3}
	\int_{0}^{\infty} v^{2}(s, T + R - s) s ds \leq E(v),
	\end{equation}
	
	\noindent and
	
	\begin{equation}\label{2.27.4}
	\int_{0}^{\infty} v^{2}(s, T - R + s) s ds \leq E(v),
	\end{equation}
	
	\noindent so therefore,
	
	\begin{equation}\label{2.27.5}
	- \frac{1}{2} \int_{\mathbb{R}} \int_{\mathbb{R}} \partial_{R} (Z_{2}^{(1)}(R, T)) v(R, T)^{2} R dR dT \leq E(v) \int_{\mathbb{R}} \int_{\mathbb{R}} v^{2}(R, T) dR dT.
	\end{equation}
	
	\noindent Now by the product rule, $u = Rv$, 
	
	\begin{equation}\label{2.27.6}
	-\frac{1}{2} (\partial_{\xi} u)(\partial_{\eta} u) = \frac{1}{2} v^{2} + \frac{1}{2} R v(\partial_{\eta} v) - \frac{1}{2} R v (\partial_{\xi} v) - \frac{1}{2} R^{2} (\partial_{\xi} v)(\partial_{\eta} v).
	\end{equation}
	
	\noindent Similar to $(\ref{2.1.1})$,
	
	\begin{equation}\label{2.27.7}
	Z_{2}^{(2)}(R, T) = -\frac{1}{2} \int_{0}^{R} \int_{T - (R - s)}^{T + (R - s)} v^{2}(s, t) dr dt,
	\end{equation}
	
	\noindent and can be analyzed in a manner similar to $(\ref{2.25})$ - $(\ref{2.27.5})$. Next, making a change of variables,
	
	\begin{align}\label{2.4.0}
	&\int_{0}^{R} \int_{T - (R - s)}^{T + (R - s)} s (\partial_{\eta} v) v dt ds = \frac{1}{2} \int_{T - R}^{T + R} \int_{T - R}^{\xi} (\xi - \eta) (\partial_{\eta} v) v d\eta d\xi \notag\\
	&= \frac{1}{4} \int_{T - R}^{T + R} \int_{T - R}^{\xi} (\xi - \eta) \partial_{\eta}(v^{2}) d\eta d\xi.
	\end{align}
	
	\noindent Integrating by parts,
	
	\begin{align}
	&= \frac{1}{4} \int_{T - R}^{T + R} (\xi - (T - R)) v^{2}(T - R, \xi) d\xi - \frac{1}{4} \int_{T - R}^{T + R} \int_{T - R}^{\xi} v^{2}(\eta, \xi) d\eta d\xi \notag\\ &= \frac{1}{4} \int_{0}^{R} s \cdot v^{2}(s, (T - R) + s) ds + \frac{1}{4} \int_{0}^{R} \int_{T - (R - s)}^{T + (R - s)} s v^{2}(s, t) dt ds \notag\\
	&= Z_{1}^{(1)}(R, T) + Z_{2}^{(3)}(R, T).
	\end{align}
	
	\noindent The contribution of $Z_{2}^{(3)}(R, T)$ may again be estimated as in $(\ref{2.25})$ - $(\ref{2.27.5})$. Now,
	
	\begin{equation}\label{2.5.1}
	Z_{1}^{(1)}(R, T) = \frac{1}{4} \int_{0}^{R} s \cdot v^{2}(s, (T - R) + s) ds \leq R \int_{\eta = T - R} v^{2}(\xi, \eta) d\xi.
	\end{equation}
	
	\noindent Then by $(\ref{1.18.3})$ and $(\ref{2.27.3})$,
	
	\begin{align}\label{2.5.2}
	&\int_{\mathbb{R}} \int_{\mathbb{R}} v_{R}(R, T) v(R, T) R Z_{1}^{(1)}(R, T) dR dT \notag\\ 
	&\leq \int \int v_{R}(\xi, \eta) R^{2} v(\xi, \eta) (\int v^{2}(\xi, \eta = T - R) d\xi) d\xi d\eta \leq E(v) \int_{\mathbb{R}} \int_{\mathbb{R}} v^{2}(R, T) dR dT.
	\end{align}

	\noindent By a similar calculation,
	
	\begin{align}
	&\int_{0}^{R} \int_{T - (R - s)}^{T + (R - s)} s (\partial_{\xi} v) v dT' ds = \frac{1}{2} \int_{T - R}^{T + R} \int_{\eta}^{T + R} (\xi - \eta) (\partial_{\xi} v) v d\xi d\eta \notag\\
	&= \frac{1}{4} \int_{T - R}^{T + R} \int_{\eta}^{T + R} (\xi - \eta) \partial_{\xi}(v^{2}) d\xi d\eta \notag\\
	&= \frac{1}{4} \int_{T - R}^{T + R} (T + R - \eta) v^{2}(\eta, \xi = T + R) d\eta - \frac{1}{4} \int_{T - R}^{T + R} \int_{\eta}^{T + R} v^{2}(\eta, \xi) d\xi d\eta. \label{2.7.0}
	\end{align}
	\noindent Then set
	
	\begin{equation}\label{2.7.1}
	Z_{1}^{(2)}(R, T) = \frac{1}{4} \int_{T - R}^{T + R} (T + R - \eta) v^{2}(\eta, \xi = T + R) d\eta
	\end{equation}
	
	\noindent and
	
	\begin{equation}\label{2.7.2}
	Z_{2}^{(4)}(R, T) = - \frac{1}{4} \int_{T - R}^{T + R} \int_{\eta}^{T + R} v^{2}(\eta, \xi) d\xi d\eta = -\frac{1}{4} \int_{0}^{R} \int_{T - (R - s)}^{T + (R - s)} v^{2}(s, T') dT' ds.
	\end{equation}
	
	\noindent It only remains to compute
	
	\begin{equation}\label{2.7.3}
	\frac{1}{2} \int_{0}^{R} \int_{T + (R - s)}^{T + (R - s)} s^{2} (\partial_{\xi} v(s, T')) (\partial_{\eta} v(s, T')) ds dT' = \frac{1}{2} \int_{T - R}^{T + R} \int_{\eta}^{T + R} (\xi - \eta)^{2} (\partial_{\xi} v)(\partial_{\eta} v) d\xi d\eta.
	\end{equation}
	
	\noindent Integrating by parts,
	
	\begin{equation}\label{2.8}
	\frac{1}{2} \int_{T - R}^{T + R} \int_{\eta}^{T + R} (\xi - \eta)^{2} (\partial_{\xi} v)(\partial_{\eta} v) d\xi d\eta = \frac{1}{2} \int_{T - R}^{T + R} (T + R - \eta)^{2} v (\partial_{\eta} v)(T + R, \eta) d\eta
	\end{equation}
	
	\begin{equation}\label{2.9}
	- \int_{T - R}^{T + R} \int_{\eta}^{T + R} (\xi - \eta) v (\partial_{\eta} v) d\xi d\eta -\frac{1}{2} \int_{T - R}^{T + R} \int_{\eta}^{T + R} (\xi - \eta)^{2} v (\partial_{\xi \eta} v) d\xi d\eta.
	\end{equation}

	\noindent By the Cauchy - Schwartz inequality,
	
	\begin{align}\label{2.11.0}
	&\frac{1}{2} \int_{T - R}^{T + R} (T + R - \eta)^{2} v (\partial_{\eta} v)(T + R, \eta) d\eta = Z_{2}^{(3)}(R, T) \notag\\ &= \frac{1}{2} \int_{0}^{R} s^{2} v(s, T + R - s) (\partial_{\eta} v(s, T + R - s)) ds \notag\\ 
	&\leq \frac{R^{1/2}}{2} \left(\int s^{3} (\partial_{\eta} v)^{2} (\xi = T + R, \eta) d\eta \right)^{1/2} \left(\int v(\xi = T + R, \eta)^{2} d\eta \right)^{1/2}.
	\end{align}
	
	\noindent Again by $(\ref{1.18.3})$ and $(\ref{2.27.3})$,
	
	\begin{align}\label{2.11.1}
	&\int_{\mathbb{R}} \int_{\mathbb{R}} Z_{2}^{(3)}(R, T) v_{R}(R, T) v(R, T) R dR dT \leq \frac{1}{2} \int_{\mathbb{R}} \int_{\mathbb{R}} (\xi - \eta)^{3/2} v_{R}(\xi, \eta) v(\xi, \eta) \notag\\
	&\leq E(v) \int_{\mathbb{R}} \int_{\mathbb{R}} v(R, T)^{2} dR dT.
	\end{align}
	
	\noindent Also observe that
	
	\begin{equation}\label{2.10.0}
	\int_{T - R}^{T + R} \int_{\eta}^{T + R} (\xi - \eta) v (\partial_{\eta} v) d\xi d\eta = (\ref{2.4.0}),
	\end{equation}
	
	\noindent may be computed in exactly the same manner as in \ref{2.4.0} splitting into $Z_{2}^{(4)} + Z_{1}^{(5)}$. Now 
	
	\begin{equation}\label{2.10.1}
	\partial_{\xi \eta} v = (\partial_{TT} - \partial_{RR}) v.
	\end{equation}

	\noindent Then 
	
	\begin{equation}\label{2.12.0}
	\aligned
	-\frac{1}{2} \int_{T - R}^{T + R} \int_{\eta}^{T + R} (\xi - \eta)^{2} v (\partial_{\xi \eta} v) d\xi d\eta = \frac{3}{2} \int_{T - R}^{T + R} \int_{\eta}^{T + R} (\xi - \eta) v(\partial_{R} v) d\xi d\eta \\ - \frac{1}{2} \int_{T - R}^{T + R} \int_{\eta}^{T + R} (\xi - \eta)^{2} v F(v) d\xi d\eta.
	\endaligned
	\end{equation}
	
	
	\begin{equation}\label{2.13.0}
	\frac{3}{2} \int_{T - R}^{T + R} \int_{\eta}^{T + R} (\xi - \eta) v (\partial_{R} v) d\xi d\eta = \frac{3}{2} \int_{0}^{R} \int_{T - (R + s)}^{T + (R - s)} v (\partial_{s} v) s ds dt,
	\end{equation}
	
	\noindent and so integrating by parts,
	
	\begin{equation}\label{2.14}
	\aligned
	= \frac{3}{4} \int_{0}^{R} v^{2}(s, T + R - s) s ds + \frac{3}{4} \int_{0}^{R} v^{2}(s, T - R + s) s ds \\ - \frac{3}{2} \int_{0}^{R} \int_{T - (R + s)}^{T + (R - s)} v^{2}(s, T') ds dT' = Z_{2}^{(5)}(R, T) + Z_{1}^{(6)}(R, T).
	\endaligned
	\end{equation}
	
	
	
	\noindent Next split $F(v) = F_{1}(v) + F_{2}(v)$ with
	
	\begin{equation}\label{2.16}
	\aligned
	F_{1}(v) = \frac{1}{R} (\partial_{\eta} r - \frac{1}{2}) (\partial_{\xi} v) + \frac{1}{R} (\partial_{\xi} r + \frac{1}{2}) (\partial_{\eta} v), \\
	F_{2}(v) =  (e^{2Z} - 1 + \partial_{\eta} r + \frac{1}{2} - (\partial_{\xi} r - \frac{1}{2})) \frac{v}{R^{2}} + e^{2Z} \frac{R^{2}}{r^{2}} v^{3} \zeta(Rv).
	\endaligned
	\end{equation}
	
	\noindent By the radial Sobolev embedding theorem, the boundedness of $\zeta(Rv)$, and 
	\[ \Big\vert \partial_{\xi} r + \frac{1}{2}\Big\vert , \Big\vert \partial_{\eta} r - \frac{1}{2}\Big\vert \quad \text{and} \quad  \vert Z \vert  \leq \epsilon(E(v)), \] 
	
	then
	
	\begin{equation}\label{2.23}
	\int_{0}^{R}  \int_{T - (R + s)}^{T + (R - s)} F_{2}(v) v s^{2} ds dT' \leq \epsilon(E(v)) \int_{\mathbb{R}} \int_{\mathbb{R}} v^{2}(R, T) dR dT.
	\end{equation}
	
	\noindent Next, integrating by parts,
	
	\begin{equation}\label{2.17}
	\int_{0}^{R} \int_{T - (R - s)}^{T + (R - s)} s (\partial_{\eta} r - \frac{1}{2}) v (\partial_{\xi} v) ds dT',
	\end{equation}
	
	\noindent may be computed in a way that is very similar to $(\ref{2.6})$, the only additional term is a term of the form
	
	\begin{equation}\label{2.18}
	-\frac{1}{2} \int_{0}^{R} \int_{T - (R - s)}^{T + (R - s)} s (\partial_{\xi \eta} r) v^{2} ds dT'.
	\end{equation}
	
	\noindent Now since $ \vert \ptl^2_{\xi \eta} r \vert \leq \frac{\epsilon(E(v))}{r} $, this term may be bounded by
	
	\begin{equation}\label{2.19}
	\int_{0}^{R} \int_{T - (R - s)}^{T + (R - s)} v^{2}(s, T') ds dT'.
	\end{equation}
	
	\noindent A similar computation may be made for
	
	\begin{equation}\label{2.20}
	\int_{0}^{R} \int_{T - (R - s)}^{T + (R - s)} s (\partial_{\xi} r - \frac{1}{2}) v (\partial_{\eta} v) ds dT'.
	\end{equation}
	
	
	
	
	

	
	
	

	
	
	\noindent Thus,
	
	\begin{equation}
	(\ref{1.15}) \leq \epsilon(E(v)) \int \int v(R, T)^{2} dR dT.
	\end{equation}
	
	\noindent Many of the computations in $(\ref{1.15})$ may also be used to estimate $(\ref{1.15.1})$. Expanding $\zeta(z)$ for $ \vert z \vert \leq \epsilon$, by the radial Sobolev embedding theorem which implies $R  \vert v(R, T) \vert \leq \epsilon(E(v))$,
	
	\begin{equation}
	\zeta(Rv) = c_{0} + c_{1} (Rv) + c_{2} (Rv)^{2} + ...
	\end{equation}
	
	\begin{equation}\label{1.16}
	\aligned
	c_{j} \int_{\mathbb{R}} \int_{\mathbb{R}}  \frac{R^{2 + j}}{r^{2}} v^{3 + j}(R, T) v_{R}(R, T) R^{3} dR dT = \frac{c_{0}}{4 + j} \int_{\mathbb{R}}  \int_{\mathbb{R}}  \frac{R^{5 + j}}{r^{2}} \partial_{R}(v(R, T)^{4 + j}) dR dT \\
	= -\frac{(5 + j)c_{j}}{4 + j} \int_{\mathbb{R}}  \int _{\mathbb{R}} \frac{R^{4 + j}}{r^{2}} v(R, T)^{4 + j} dR dT - \frac{c_{j}}{2} \int_{\mathbb{R}}  \int_{\mathbb{R}}  \frac{R^{5 + j}}{r^{3}} (\partial_{R} r) v(R, T)^{4 + j} dR dT.
	\endaligned
	\end{equation}
	
	\noindent Therefore,
	
	\begin{equation}\label{1.17.0}
	\int_{\mathbb{R}} \int_{\mathbb{R}}  \frac{R^{2}}{r^{2}} v(R, T)^{3} v_{R}(R, T) R^{3} dR dT \lesssim \epsilon(E(v)) \int_{\mathbb{R}}  \int_{\mathbb{R}}  v(R, T)^{2} dR dT.
	\end{equation}
	
	\noindent Next, we estimate
	
	\begin{equation}\label{1.17.1}
	c_{j} \int_{\mathbb{R}}  \int_{\mathbb{R}}  \frac{R^{2 + j}}{r^{2}} Z(R, T) v^{3 + j}(R, T) v_{R}(R, T) R^{3} dR dT.
	\end{equation}
	
	\noindent Again splitting $Z(R, T) = Z_{1}(R, T) + Z_{2}(R, T)$,
	
	\begin{equation}\label{1.17.2}
	\aligned
	\frac{c_{0}}{4 + j} \int_{\mathbb{R}} \int_{\mathbb{R}}  \frac{R^{5 + j}}{r^{2}} Z_{1}(R, T) v(R, T)^{3 + j} \partial_{R} v(R, T) dR dT \\ = \frac{c_{0}}{4 + j} \int_{\mathbb{R}}  \int_{\mathbb{R}}  \frac{R^{5 + j}}{r^{2}} Z_{1}(R, T) \partial_{R}(v(R, T)^{4 + j}) dR dT \\ = -\frac{(5 + j)c_{j}}{4 + j} \int_{\mathbb{R}}  \int_{\mathbb{R}}  \frac{R^{4 + j}}{r^{2}} Z_{1}(R, T) v(R, T)^{4 + j} dR dT \\ - \frac{c_{j}}{2} \int_{\mathbb{R}}  \int_{\mathbb{R}}  \frac{R^{5 + j}}{r^{3}} Z_{1}(R, T) (\partial_{R} r) v(R, T)^{4 + j} dR dT \\ - \frac{c_{j}}{4 + j} \int_{\mathbb{R}}  \int_{\mathbb{R}}  \frac{R^{5 + j}}{r^{2}} (\partial_{R} Z_{1}(R, T)) v(R, T)^{4 + j} dR dT.
	\endaligned
	\end{equation}
	
	\noindent Then by the radial Sobolev embedding theorem, properties of $r$, and the analysis of $(\ref{1.15.3})$ implies that
	
	\begin{equation}\label{1.17.3}
	(\ref{1.17.2}) \leq \epsilon(E(v)) \int_{\mathbb{R}}  \int_{\mathbb{R}}  v(R, T)^{2} dR dT.
	\end{equation}
	
	\noindent The radial Sobolev embedding theorem, properties of $r$, and the analysis of $(\ref{1.15.3})$ also implies that
	
	\begin{equation}\label{1.17.4}
	\frac{c_{0}}{4 + j} \int_{\mathbb{R}}  \int_{\mathbb{R}}  \frac{R^{5 + j}}{r^{2}} Z_{2}(R, T) v(R, T)^{3 + j} \partial_{R} v(R, T) dR dT \leq \epsilon(E(v)) \int_{\mathbb{R}}  \int_{\mathbb{R}}  v(R, T)^{2} dR dT.
	\end{equation}
	
	\noindent Therefore, we have now proved
	
	\begin{equation}\label{1.17.5}
	\int_{\mathbb{R}}  \int_{\mathbb{R}}  v(R, T)^{2} dR dT \leq E(v) + \epsilon(E(v)) \int_{\mathbb{R}}  \int_{\mathbb{R}}  v(R, T)^{2} dR dT.
	\end{equation}
	
	\noindent This completes the proof of Theorem $\ref{t1.1}$. $\Box$
	
\end{proof}

\begin{theorem}
	If $v$ is a global solution of \eqref{wave-full} for $E_0 < \eps^2$, then
	\begin{align}
	v \to v_S, \quad \textnormal{as} \quad T \to \infty,
	\end{align}
 in the energy topology, $\Vert \cdot \Vert_{L^2 (\mathbb{R}^4) \times \dot{H}^1 (\mathbb{R}^4 ) }, $
	where $v_S$ is a solution of linearized \eqref{eq:waveuTR}
	\begin{align}
	\leftexp{4+1}{\square} v_S =
	-\ptl_T^2 v_S +\ptl_{R}^2 v_S +\frac{3}{R}\ptl_{R} v_S =0
	\end{align}
	which in turn is equivalent to linearized \eqref{eq:waveuTR}
	\begin{align}
	-\ptl_T^2 u_S +\ptl_{R}^2 u_S +\frac{1}{R}\ptl_R u_S -\frac{1}{R^2} u_S =0
	\end{align}
 for $u_S \fdg = R v_S.$
\end{theorem}
\begin{proof}
The methodology of the proof is similar to that of Proof of Problem II in \cite{BN_17}, but it now extends to the full wave maps equation \eqref{wave-full} (recall that, $u=Rv$).

	Consider a large time $T_0$ and let $v =\bar{v} + v_S$, such that $\bar{v}$ and $v_S$ are solutions of 
	\begin{equation}\label{barv}
	\left. \begin{array}{rcl}
	\leftexp{4+1}{\square}\, \bar{v} &=& F(v)\,\,\,\,\,\,\,\,\,\,\,\,\,\,\, \,\,\,\,\,\,\,\,\,\,\,\,\,\,\,\,\,\,\,\,\,\,\,\,\,\,\,\,\,\textnormal{on}\,\, \mathbb{R}^{4+1}\\
	\bar{v}_0 = \bar{v} (T_0, x) =0 & \textnormal{and}& \bar{v}_1 = \ptl_T \bar{v} (T_0, x)=0 \,\,\,\,\,\,\,\, \textnormal{on}\,\, \mathbb{R}^4\\\end{array} 
	\right\}
	\end{equation}
	where the forcing term
	\begin{align}
	F(v) \fdg=& \left(e^{2Z} - 1 + \left(\frac{r}{R}\partial_{\eta} r + \frac{1}{2}\right) - \left(\frac{r}{R}\partial_{\xi} r - \frac{1}{2}\right)\right) \frac{v}{r^{2}}  + 2 \ptl_\xi v \ptl_\eta \log \left(\frac{r}{R}\right) \notag\\ &+ 2 \ptl_\eta v \ptl_\xi \log\left(\frac{r}{R}\right)  + e^{2Z} \frac{R^{2}}{r^{2}} v^{3} \zeta(R v)
	\end{align}
	and \begin{equation}\label{vS}
	\left. \begin{array}{rcl}
	\leftexp{4+1}{\square}\, v_S &=& 0\,\,\,\,\,\,\,\,\,\,\,\,\,\,\, \,\,\,\,\,\,\,\,\,\,\,\,\,\,\,\,\,\,\,\,\,\,\,\,\,\,\,\,\,\,\,\,\,\,\,\textnormal{on}\,\, \mathbb{R}^{4+1}\\
	(v_S)_0 = v (T_0, x) & \textnormal{and}& (v_S)_1 = \ptl_T v (T_0, x) \,\,\,\,\,\,\,\, \textnormal{on}\,\, \mathbb{R}^4\\\end{array} 
	\right\}
	\end{equation}
	so that $v(T) = \bar{v}(T) + S(T-T_0) (v_0, v_1)$
	
	where $S(T) (v_0,v_1)$ is the linear evolution operator of the wave equation \eqref{vS}. It follows from the triangle inequality that $\bar{E}(\bar{v}) \leq E(v)$, where 
	\begin{align}
	\bar{E} (\bar{v}) =& \Vert \ptl_T \bar{v} \Vert^2_{L^2 (\mathbb{R}^4)} + \Vert \grad \bar{v} \Vert^2_{L^2 (\mathbb{R}^4)} + 
	\halb \Vert \bar{v} \Vert^4_{L^4 (\mathbb{R}^4)}  
 \intertext{and}
 	\bar{E} (v) =& \Vert \ptl_T v \Vert^2_{L^2 (\mathbb{R}^4)} + \Vert \grad v \Vert^2_{L^2 (\mathbb{R}^4)} + 
	\halb \Vert v \Vert^4_{L^4 (\mathbb{R}^4)}  .
	\end{align}
 
 We have the identity, 

 \begin{align} \label{instiden}
	\frac{\ptl}{\ptl T} \left( \halb \ip{\ptl_T \bar{v}}{\ptl_T \bar{v}} + \halb \ip{\grad \bar{v}}{\grad \bar{v}}\right) =  - \ip{\ptl_T \bar{v}}{F (v)}
	\end{align}
	where $\ip{X}{Y} =  \int X \cdot Y \, \bar{\mu}_{\check{g}},$ which is an instantaneous version of the energy identity: 
 \begin{align}
	\int_{\Sigma_{T_1}} \check{e} (\bar{v}) \bar{\mu}_{\check{q}} - \int_{\Sigma_{T_2}} \check{e} (\bar{v}) \bar{\mu}_{\check{q}} = \int  \ptl_T \bar{v} F(v)\bar{\mu}_{\check{g}}
	\end{align}
  where $ \check{e} (\bar{v}) = \halb \vert \ptl_T \bar{v} \vert^2 + \halb \vert \grad \bar{v} \vert^2. $

 Now consider the terms on the right hand side of \eqref{instiden} consecutively,
	
	\begin{enumerate}
		
		\item Firstly consider the terms \[ \int^{\infty}_{T_0} \int_{\mathbb{R}^4} \frac{1}{R} (\ptl_\eta r + \halb) \ptl_\xi v \ptl_T \bar{v} \bar{\mu}_{\check{g}}\,\, \text{and}\,\, \int^\infty_{T_0} \int_{\mathbb{R}^4} \frac{1}{R} (\ptl_\xi r - \halb) \ptl_\eta v \ptl_ T\bar{v} \bar{\mu}_{\check{g}} :\] 
		we have
		\begin{align}
		&\int_{T_0} \int_{\mathbb{R}^4} \frac{1}{R} (\ptl_\eta r + \halb) \ptl_\xi v \ptl_ T \bar{v} R^3 dR dT \\
		&=  \int_{T_0} \int_{\mathbb{R}^4} \frac{1}{R} (\ptl_\eta r + \halb) \ptl_\xi v \ptl_ T \bar{v} R^3 d\xi \eta \\
		& \leq \int_{\eta_0} (\sup_{R >0} \frac{1}{R} \int^R_{0} f^2(v)sds) (\int (dv)^2 d\xi) d\eta \\
		\intertext{going back to $(T, R, \theta)$}
		&\leq E(v) \int^\infty_{T_0} \int_{\mathbb{R}^4} \frac{v^2}{R^3} \bar{\mu}_{\check{g}}
		\end{align} 
		analogously, we have
		\begin{align}
		\int^{\infty}_{T_0} \int_{\mathbb{R}^4} \frac{1}{R} (\ptl_\xi r - \halb) (\ptl_\eta v) (\ptl_T v) \bar{\mu}_{\check{g}} 
		\leq E(v) \int^\infty_{T_0} \int_{\mathbb{R}^4} \frac{v^2}{R^3} \bar{\mu}_{\check{g}}
		\end{align}
		\item Next, the terms \[\int^\infty_{T_0} \int_{\mathbb{R}^4} (\ptl_\eta r + \halb) \frac{v}{R^2} \ptl_T \bar{v} \bar{\mu}_{\check{g}} \,\, \text{and} \,\, \int^\infty_{T_0} \int_{\mathbb{R}^4} (\ptl_\xi r - \halb) \frac{v}{R^2} \ptl_T \bar{v} \bar{\mu}_{\check{g}}\]
		using Cauchy-Schwarz 
		\begin{align}
		&\int^\infty_{T_0} \int_{\mathbb{R}^4} (\ptl_\eta r + \halb) \frac{v}{R^2} \ptl_T \bar{v} \bar{\mu}_{\check{g}}  \notag\\
		&\leq \int^\infty_{T_0} \int_{\mathbb{R}^4} \frac{1}{R} (\ptl_\eta r + \halb) (\ptl_T \bar{v})^2  \bar{\mu}_{\check{g}} 
		+ \int^\infty_{T_0} \int_{\mathbb{R}^4} \frac{1}{R^3} (\ptl_\eta r + \halb) v^2 \bar{\mu}_{\check{g}} \notag\\
		&\leq E(v) \int^\infty_{T_0} \int_{\mathbb{R}^4} \frac{v^2}{R^3} \bar{\mu}_{\check{g}} + \eps \int^\infty_{T_0} \int_{\mathbb{R}^4} \frac{v^2}{R^3} \bar{\mu}_{\check{g}}
		\end{align}
		\noindent likewise, we have
		\begin{align*}
		\int^\infty_{T_0} \int_{\mathbb{R}^4} (\ptl_\xi r - \halb) \frac{v}{R^2} \ptl_T \bar{v} \bar{\mu}_{\check{g}} \leq 
		E(v) \int^\infty_{T_0} \int_{\mathbb{R}^4} \frac{v^2}{R^3} \bar{\mu}_{\check{g}} + \eps \int^\infty_{T_0} \int_{\mathbb{R}^4} \frac{v^2}{R^3} \bar{\mu}_{\check{g}}
		\end{align*}
		\item Controlling the term
		\begin{align}
		\int^\infty_{T_0} \int_{\mathbb{R}^4} (e^{2Z} -1) \frac{v}{R^2} \ptl_T \bar{v} \bar{\mu}_{\check{g}} 
		\end{align}
		is similar,
		\begin{align}
		\int^\infty_{T_0} \int_{\mathbb{R}^4} (e^{2Z} -1) \frac{v}{R^2} \ptl_T \bar{v} \bar{\mu}_{\check{g}}&\leq \int^\infty_{T_0} \int_{\mathbb{R}^4}\frac{1}{R} (e^{2Z} -1) \bar{\mu}_{\check{g}} + \int^\infty_{T_0} \int_{\mathbb{R}^4} (e^{2Z} -1) \frac{v^2}{R^3} \bar{\mu}_{\check{g}} \notag\\
		& \leq  E(v) \int^\infty_{T_0} \int_{\mathbb{R}^4} \frac{v^2}{R^3} \bar{\mu}_{\check{g}} + \eps \int^\infty_{T_0} \int_{\mathbb{R}^4} \frac{v^2}{R^3} \bar{\mu}_{\check{g}}.
		\end{align}
		In view of the fact that 
		\[\lim_{T_0 \to \infty} \int^\infty_{T_0} \int_{\mathbb{R}^4} \frac{v^2}{R^3} \bar{\mu}_{\check{g}} =0 \] 
		we now have
		\begin{subequations}
			\begin{align}
			\lim_{T_0 \to \infty } \int^\infty_{T_0} \int_{\mathbb{R}^4} \frac{1}{R} (\ptl_\eta r + \halb) \ptl_\xi v \ptl_T\bar{v} \bar{\mu}_{\check{g}} =& 0 \\
			\lim_{T_0 \to \infty } \int^\infty_{T_0} \int_{\mathbb{R}^4} \frac{1}{R} (\ptl_\xi r - \halb) \ptl_\eta v \ptl_T\bar{v} \bar{\mu}_{\check{g}} =& 0 \\
			\lim_{T_0 \to \infty} \int^\infty_{T_0} \int_{\mathbb{R}^4} \frac{1}{R} (\ptl_\eta r + \halb) \frac{v}{R^2} \ptl_T\bar{v} \bar{\mu}_{\check{g}} =& 0 \\
			\lim_{T_0 \to \infty} \int^\infty_{T_0} \int_{\mathbb{R}^4} \frac{1}{R} (\ptl_\xi r - \halb) \frac{v}{R^2}\ptl_T\bar{v} \bar{\mu}_{\check{g}} =& 0 \\
			\lim_{T_0 \to \infty}\int^\infty_{T_0} \int_{\mathbb{R}^4} (e^{2Z} -1) \frac{v}{R^2} \ptl_T \bar{v} \bar{\mu}_{\check{g}} =&0
			\end{align}
		\end{subequations}
		from the dominated convergence theorem. 
	
		\item Now consider the term 
		\[ \int^\infty_{T_0} \int_{\mathbb{R}^4} \frac{R^2}{r^2} e^{2Z} v^3 \zeta(Rv) \ptl_T \bar{v} \bar{\mu}_{\check{g}} \]
		
		\noindent again expand out $\zeta(R\tilde{v}):$
		\begin{align}
		\zeta (Rv) = c_0 + c_1 Rv + c_2 (Rv)^2+ \cdots
		\end{align}
		
		\noindent Then the leading order term of 
		\begin{align}
		c_j \int^\infty_{T_0} \int_{\mathbb{R}^4} \frac{R^{2+j}}{r^2} e^{2Z} v^{3+j} \ptl_T \bar{v} \bar{\mu}_{\check{g}}
		\end{align}
		transforms as 
		\begin{equation}\label{1.31}
		c_{0} \int \frac{R^{2}}{r^{2}} \bar{v}^{3} \bar{v}_{T} R^{3} \bar{\mu}_{\check{g}} = \left(-\frac{5 c_{0}}{4} \int \frac{R^{2}}{r^{2}} \bar{v}^{4} \bar{\mu}_{\check{g}} \right) - \frac{c_{0}}{2} \int \frac{R^{2}}{r^{3}} \bar{v}^{4}  (\partial_{T} r) \bar{\mu}_{\check{g}}
		\end{equation}

		\noindent Now by the radial Sobolev embedding theorem, combined with the Morawetz estimates,
		
		\begin{equation}\label{1.32}
		\int_{T \geq T_{0}} \int \frac{R^{2}}{r^{3}} \bar{v}^{4}  (\partial_{T} r) \bar{\mu}_{\check{g}} \rightarrow 0
		\end{equation}
		
		\noindent as $T_{0} \rightarrow \infty$. Next, using the radial Strichartz estimate
		
		\begin{equation}
		\Big\Vert \vert x\vert^{1/2} S(t)((v_S)_0, (v_S)_1) \Big\Vert_{L_{t}^{2} L_{x}^{\infty}} \leq \Vert (v_S)_0  \Vert_{\dot{H}^{1}} + \Vert (v_S)_1\Vert_{L^{2}},
		\end{equation}
		
		\noindent so
		
		\begin{align}\label{1.33}
		\int_{T \geq T_{0}} \int \frac{R^{2}}{r^{2}} v^{2} v_S \bar{v}_{T} R^{3} dR dT \leq& \left(\int_{T \geq T_{0}} \int v^{2} dR dT \right)^{1/2} \notag     \\
		& \quad \cdot\| R^{1/2} v_S \|_{L_{T}^{2} L_{x}^{\infty}} \| \bar{v}_{T} \|_{L_{T}^{\infty} L_{x}^{2}} \| R v \|_{L_{T, x}^{\infty}}.
		\end{align}
		
		\noindent If 
		
		\begin{equation}\label{1.34}
		E(\bar{v}(T)) = \| \bar{v}_{T} \|_{L^{2}}^{2} + \| \nabla \bar{v}(T) \|_{L^{2}}^{2} + \halb \| \bar{v}(T) \|_{L^{4}}^{4},
		\end{equation}
		
		\noindent then the Sobolev embedding theorem implies that for small energy, the RHS of eq. \ref{1.33}
		
		\begin{equation}\label{1.35}
		\leq (\sup_{T \geq T_{0}} E(\bar{v}(T)))^{1/2} E (v)\left(\int_{T \geq T_{0}} \int v^{2} dR dT \right)^{1/2}.
		\end{equation}
		
		\noindent Also,
		
		\begin{equation}\label{1.36}
		\aligned
		\int_{T \geq T_{0}} \int \frac{R^{2}}{r^{2}} v \cdot v_S \cdot \bar{v} \bar{v}_{T} R^{3} dR dT \leq & \left(\int_{T \geq T_{0}} \int v^{2} dR dT \right)^{1/2} \\
		& \cdot \| R^{1/2} v_S \|_{L_{T}^{2} L_{x}^{\infty}} \| \bar{v}_{T} \|_{L_{T}^{\infty} L_{x}^{2}} \| R \bar{v} \|_{L_{T, x}^{\infty}} \\ 
		\leq & E(v)\left(\sup_{T \geq T_{0}} E(\bar{v}(T)) \right).
		\endaligned
		\end{equation}
		
		\noindent Finally,
		
		\begin{equation}\label{1.37}
		\aligned
		\int_{T \geq T_{0}} \int \frac{R^{2}}{r^{2}} v_S \bar{v}^{2} \bar{v}_{T} R^{3} dR dT \leq & \Big\Vert  \vert x \vert^{1/2} v_S \Big\Vert_{L_{T}^{2} L_{x}^{\infty}} \| v_{T} \|_{L_{T}^{\infty} L_{x}^{2}} \\
		& \quad \cdot \left(\int v + (v_S)^{2} dR dT \right)^{1/2} \| \vert x \vert \bar{v} \|_{L_{T, x}^{\infty}} \\ 
		\leq &E(v)\left(\sup_{T \geq T_{0}} E(\bar{v})(T)\right).
		\endaligned
		\end{equation}
		
		Now consider, 
		\begin{align}
		c_j\int_{T_0} \int_{\mathbb{R}^4} \frac{R^{2+j}}{r^2} Z(R, T) v^{3+j} \ptl_T \bar{v} \bar{\mu}_{\check{g}}
		\end{align}
Analogously, we have		
		\begin{align}
		c_j\int_{T_0} \int_{\mathbb{R}^4} \frac{R^{2+j}}{r^2} Z(R, T) v^{3+j} \ptl_T \bar{v} \bar{\mu}_{\check{g}} \leq c(E_0) \int v^2 dR dT 
		\end{align}

	\end{enumerate} 
\end{proof}
In summary, we considered the $2+1$ dimensional geometric equivariant wave maps equation coupled to the Einstein equations for general relativity. We reduced the geometric $2+1$ dimensional equivariant wave maps equation into a $4+1$ dimensional non-linear wave equation on the Minkowski space. This allows us to represent the solution into two parts: solutions of a homogeneous wave equation with the same initial data $v_S$ and a non-homogeneous wave equation with zero data $\bar{v}$. We then show that the energy of the latter goes to $0$ as $T \to \infty$: 

\begin{align}
    \sup_{T \geq T_0}  \left( \Vert  \ptl_T \bar{v} \Vert_{L^2 (\mathbb{R}^4)} + \Vert  \bar{v} \Vert_{\dot{H}^1 (\mathbb{R}^4)}  \right ) \to 0
\end{align}
 as $T_0 \to \infty.$ In these estimates, the non-linear Morawetz estimate \eqref{morawetz-estimate} plays an important role. We would now like to remark on the decay rates for our scattering problem in comparison with the decay rates for the linear wave equation. 
\begin{remark}
	Suppose $ \phi  \fdg \mathbb{R}^{2+1} \to \mathbb{R}$ satisfies the linear wave equation
	\begin{equation}
	\left. \begin{array}{rcl}
	\leftexp{2+1}{\square} \, \phi &=&0\,\,\,\,\,\,\,\,\,\,\,\,\,\,\, \,\,\,\,\,\,\,\,\,\,\,\,\,\,\,\,\,\,\,\,\,\,\,\,\,\,\,\text{on}\,\, \mathbb{R}^{2+1}\\
	\phi_0 = \phi (0, x) & \text{and}& \phi_1 = \ptl_T \phi (0, x) \,\,\,\,\,\,\,\, \text{on}\,\, \mathbb{R}^2\\\end{array} 
	\right\}
	\end{equation}
	such that the smooth initial data $(\phi_0, \phi_1)$ is compactly supported or sufficiently rapidly decaying, then $\phi$ admits the decay rate: 
	
	\begin{align} \label{2+1-decay}
	\vert \phi \vert \leq (1+\xi)^{-1/2} (1+\eta)^{-1/2}.
	\end{align}
\end{remark}
This decay rate has to be contrasted with the decay rate from the (standard) energy estimates:

\begin{align}
\vert \phi \vert \leq \left(1+t \right)^{- \left(\frac{n-1}{2} \right)} \left(1+ \vert t- \vert x \vert  \vert \right)^{1/2}.
\end{align}
In the following we shall prove that this decay rate is also also consistent with the decay rate for the fully coupled 2+1 equivariant Einstein-wave map system.  

\begin{corollary}
	Suppose $(M, g, U)$ is the smooth, globally hyperbolic, geodesically complete maximal development of the smooth, asymptotically flat initial data $ (\Sigma_0, q_0, k_0, U_0),$ with $U_0$ compactly supported, then $u$ admits the decay rate
	\begin{align}
	\vert u \vert \leq (1+ \xi)^{-1/2} (1+ \eta)^{-1/2} 
	\end{align}
	in the temporal asymptotic region $(T \to \infty)$ and
	\begin{align}
	\vert u \vert = \mathscr{O}(R^{-1/2})
	\end{align}
	as $R \to \infty$ along the $\eta = const.$ hypersurfaces. 
\end{corollary}
\begin{proof}
	From our scattering results, we have 
	\begin{align}
 \Vert v - v_S \Vert \to 0, \quad \textnormal{as} \quad  T \to \infty
 \end{align}
 in the energy topology. From the radial Sobolev embedding, this implies
	
 \begin{align} \label{u-Linfty-estimate}
 \Vert u - u_S \Vert_{L^\infty} \to 0, \quad T \to \infty, 
 \end{align}
 where $u = Rv$ and $u_S = R v_S.$
	Recall that if $v_S$ satisfies $\leftexp{4+1}{\square} v_S =0$ then $u_S$ satisfies 
	\begin{align}
	\leftexp{2+1}{\square} u_S - \frac{u_S}{R^2} =0.
	\end{align}
	Following \cite{chris_tah1}, let us introduce the following `projection' formulas 
	\begin{align}
	u_1 = u_S \cos \theta \notag\\
	u_2 = u_S \sin \theta
	\end{align} 
 so that $u_S = u_1 \cos \theta + u_2 \sin \theta. $ 
Now consider the wave operator $\leftexp{2+1}{ \square} u_1$. We have 

 \begin{align}
     \leftexp{2+1}{ \square} u_1 =& - \frac{\ptl^2 u_1 }{ \ptl T^2} 
     + \frac{\ptl^2 u_1 }{ \ptl R^2}  + \frac{1}{R} \frac{\ptl u_1 }{ \ptl R} + \frac{1}{R^2} \frac{\ptl^2 u_1 }{ \ptl \theta ^2} \notag \\
     =& \cos \theta \left (  - \frac{\ptl^2 u_S  }{ \ptl T^2} +
     \frac{\ptl^2 u_S  }{ \ptl R^2} + \frac{1}{R} \frac{\ptl u_S  }{ \ptl R } - \frac{u_S}{R^2} \right)  \notag \\
     =& \cos \theta \left ( \leftexp{2+1}{ \square} u_S - \frac{u_S}{R^2} \right) =0. \notag \\     
 \end{align}
An analogous computation holds for $u_2.$ We then have,	
	\begin{align}
	\leftexp{2+1}{\square} u_1 =&\,0, \quad \mathbb{R}^{2+1} \\
	\leftexp{2+1}{\square} u_2 =&\,0, \quad \mathbb{R}^{2+1}.
	\end{align}
	If we turn to the 2+1 decay rates, we have from \eqref{2+1-decay}
	\begin{align}
	 \vert u_1 \vert  \leq  (1+ \xi)^{-1/2} (1+ \eta)^{-1/2} \notag\\
\vert 	u_2 \vert  \leq  (1+ \xi)^{-1/2} (1+ \eta)^{-1/2} 
	\end{align}
	which, in view of the previous equations, yields
	\begin{align}
\vert 	u_S  \vert \leq (1+ \xi)^{-1/2} (1+ \eta)^{-1/2}.
	\end{align}
	\end{proof}

\bmhead{Competing Interests and Funding} 
The author declares that there no competing interests.  

\bibliography{central-bib(Harvard).bib}


\begin{thebibliography}{23}
\ifx \bisbn   \undefined \def \bisbn  #1{ISBN #1}\fi
\ifx \binits  \undefined \def \binits#1{#1}\fi
\ifx \bauthor  \undefined \def \bauthor#1{#1}\fi
\ifx \batitle  \undefined \def \batitle#1{#1}\fi
\ifx \bjtitle  \undefined \def \bjtitle#1{#1}\fi
\ifx \bvolume  \undefined \def \bvolume#1{\textbf{#1}}\fi
\ifx \byear  \undefined \def \byear#1{#1}\fi
\ifx \bissue  \undefined \def \bissue#1{#1}\fi
\ifx \bfpage  \undefined \def \bfpage#1{#1}\fi
\ifx \blpage  \undefined \def \blpage #1{#1}\fi
\ifx \burl  \undefined \def \burl#1{\textsf{#1}}\fi
\ifx \doiurl  \undefined \def \doiurl#1{\url{https://doi.org/#1}}\fi
\ifx \betal  \undefined \def \betal{\textit{et al.}}\fi
\ifx \binstitute  \undefined \def \binstitute#1{#1}\fi
\ifx \binstitutionaled  \undefined \def \binstitutionaled#1{#1}\fi
\ifx \bctitle  \undefined \def \bctitle#1{#1}\fi
\ifx \beditor  \undefined \def \beditor#1{#1}\fi
\ifx \bpublisher  \undefined \def \bpublisher#1{#1}\fi
\ifx \bbtitle  \undefined \def \bbtitle#1{#1}\fi
\ifx \bedition  \undefined \def \bedition#1{#1}\fi
\ifx \bseriesno  \undefined \def \bseriesno#1{#1}\fi
\ifx \blocation  \undefined \def \blocation#1{#1}\fi
\ifx \bsertitle  \undefined \def \bsertitle#1{#1}\fi
\ifx \bsnm \undefined \def \bsnm#1{#1}\fi
\ifx \bsuffix \undefined \def \bsuffix#1{#1}\fi
\ifx \bparticle \undefined \def \bparticle#1{#1}\fi
\ifx \barticle \undefined \def \barticle#1{#1}\fi
\bibcommenthead
\ifx \bconfdate \undefined \def \bconfdate #1{#1}\fi
\ifx \botherref \undefined \def \botherref #1{#1}\fi
\ifx \url \undefined \def \url#1{\textsf{#1}}\fi
\ifx \bchapter \undefined \def \bchapter#1{#1}\fi
\ifx \bbook \undefined \def \bbook#1{#1}\fi
\ifx \bcomment \undefined \def \bcomment#1{#1}\fi
\ifx \oauthor \undefined \def \oauthor#1{#1}\fi
\ifx \citeauthoryear \undefined \def \citeauthoryear#1{#1}\fi
\ifx \endbibitem  \undefined \def \endbibitem {}\fi
\ifx \bconflocation  \undefined \def \bconflocation#1{#1}\fi
\ifx \arxivurl  \undefined \def \arxivurl#1{\textsf{#1}}\fi
\csname PreBibitemsHook\endcsname

\bibitem{BN_17}
\begin{barticle}
\bauthor{\bsnm{Dodson}, \binits{B.}},
\bauthor{\bsnm{Gudapati}, \binits{N.}}:
\batitle{On scattering for small data of 2+1 dimensional equivariant
  {E}instein-wave map system}.
\bjtitle{Ann. Henri. Poincar\'e}
\bvolume{18}(\bissue{9}),
\bfpage{3097}--\blpage{3142}
(\byear{2017})
\end{barticle}
\endbibitem

\bibitem{Bruhat_Geroch_classic}
\begin{barticle}
\bauthor{\bsnm{Choquet-Bruhat}, \binits{Y.}},
\bauthor{\bsnm{Geroch}, \binits{R.}}:
\batitle{Global aspects of the {C}auchy problem in general relativity}.
\bjtitle{Commun. math. Phys.}
\bvolume{14},
\bfpage{329}--\blpage{335}
(\byear{1969})
\end{barticle}
\endbibitem

\bibitem{ash_var}
\begin{barticle}
\bauthor{\bsnm{Ashtekar}, \binits{A.}},
\bauthor{\bsnm{Varadarajan}, \binits{M.}}:
\batitle{Striking property of the gravitational {H}amiltonion}.
\bjtitle{Phys. Rev. D}
\bvolume{50}(\bissue{8}),
\bfpage{4944}--\blpage{4956}
(\byear{1994})
\end{barticle}
\endbibitem

\bibitem{AGS_15}
\begin{botherref}
\oauthor{\bsnm{Andersson}, \binits{L.}},
\oauthor{\bsnm{Gudapati}, \binits{N.}},
\oauthor{\bsnm{Szeftel}, \binits{J.}}:
Global regularity for 2+1 dimensional {E}instein-wave map system.
Ann. PDE
\textbf{3}(13)
(2017)
\end{botherref}
\endbibitem

\bibitem{diss}
\begin{botherref}
\oauthor{\bsnm{Gudapati}, \binits{N.}}:
The {C}auchy problem for energy critical self-gravitating wave maps.
Dissertation (FU Berlin)
(2013)
\end{botherref}
\endbibitem

\bibitem{diss_13}
\begin{botherref}
\oauthor{\bsnm{Gudapati}, \binits{N.}}:
The {C}auchy problem for energy critical self-gravitating wave maps.
Dissertation (FU Berlin)
(2013)
\end{botherref}
\endbibitem

\bibitem{chris_tah1}
\begin{barticle}
\bauthor{\bsnm{Christodoulou}, \binits{D.}},
\bauthor{\bsnm{Tahvildar-Zadeh}, \binits{A.S.}}:
\batitle{On the regularity of spherically symmetric wave maps}.
\bjtitle{Comm. Pure Appl. Math.}
\bvolume{46}(\bissue{7}),
\bfpage{1041}--\blpage{1091}
(\byear{1993})
\end{barticle}
\endbibitem

\bibitem{chris_tah2}
\begin{barticle}
\bauthor{\bsnm{Christodoulou}, \binits{D.}},
\bauthor{\bsnm{Tahvildar-Zadeh}, \binits{A.S.}}:
\batitle{On the asymptotic behavior of spherically symmetric wave maps}.
\bjtitle{Duke Math. J.}
\bvolume{71}(\bissue{1}),
\bfpage{31}--\blpage{69}
(\byear{1993})
\end{barticle}
\endbibitem

\bibitem{jal_tah}
\begin{barticle}
\bauthor{\bsnm{Shatah}, \binits{J.}},
\bauthor{\bsnm{Tahvildar-Zadeh}, \binits{A.S.}}:
\batitle{Cauchy problem for equivariant wave maps}.
\bjtitle{Comm. Pure Appl. Math.}
\bvolume{47}(\bissue{5}),
\bfpage{719}--\blpage{754}
(\byear{1994})
\end{barticle}
\endbibitem

\bibitem{jal_tah1}
\begin{barticle}
\bauthor{\bsnm{Shatah}, \binits{J.}},
\bauthor{\bsnm{Tahvildar-Zadeh}, \binits{A.S.}}:
\batitle{Regularity of harmonic maps from the {M}inkowski space into
  rotationally symmetric manifolds}.
\bjtitle{Comm. Pure Appl. Math.}
\bvolume{45},
\bfpage{1041}--\blpage{1091}
(\byear{1992})
\end{barticle}
\endbibitem

\bibitem{grillakis}
\begin{botherref}
\oauthor{\bsnm{Grillakis}, \binits{M.}}:
Classical solutions for the equivariant wave map in 1+2 dimensions.
preprint
(1991)
\end{botherref}
\endbibitem

\bibitem{hu_16}
\begin{botherref}
\oauthor{\bsnm{Huneau}, \binits{C.}}:
Stability in exponential time of {M}inkowski space–time with a translation
  space-like {K}illing field.
Ann. PDE
\textbf{2}(7)
(2016)
\end{botherref}
\endbibitem

\bibitem{LR_10}
\begin{barticle}
\bauthor{\bsnm{Lindblad}, \binits{H.}},
\bauthor{\bsnm{Rodnianski}, \binits{I.}}:
\batitle{The global stability of {M}inkowski space-time in harmonic gauge}.
\bjtitle{Ann. Math.}
\bvolume{171}(\bissue{3}),
\bfpage{1401}--\blpage{1477}
(\byear{2010})
\end{barticle}
\endbibitem

\bibitem{CK_94}
\begin{botherref}
\oauthor{\bsnm{Christodoulou}, \binits{D.}},
\oauthor{\bsnm{Klainerman}, \binits{S.}}:
The global nonlinear stability of the minkowski space (pms-41).
Princeton University Press
(1994)
\end{botherref}
\endbibitem

\bibitem{HL_17}
\begin{barticle}
\bauthor{\bsnm{Lindblad}, \binits{H.}}:
\batitle{On the asymptotic behavior of solutions of einstein vacuum equations
  in wave coordinates}.
\bjtitle{Commun. Math. Phys.}
\bvolume{353}(\bissue{1}),
\bfpage{135}--\blpage{184}
(\byear{2017})
\end{barticle}
\endbibitem

\bibitem{tao_all}
\begin{botherref}
\oauthor{\bsnm{Tao}, \binits{T.}}:
Global regularity of wave maps i to viii.
arXiv
(2001-09)
\end{botherref}
\endbibitem

\bibitem{tat_besovl}
\begin{barticle}
\bauthor{\bsnm{Tataru}, \binits{D.}}:
\batitle{On global existence and scattering for the wave maps equation}.
\bjtitle{Amer. J. Math.}
\bvolume{123}(\bissue{3}),
\bfpage{385}--\blpage{423}
(\byear{2001})
\end{barticle}
\endbibitem

\bibitem{tat_besovh}
\begin{barticle}
\bauthor{\bsnm{Tataru}, \binits{D.}}:
\batitle{Local and global results for wave maps {I}}.
\bjtitle{Comm. PDE}
\bvolume{23}(\bissue{9-10}),
\bfpage{1781}--\blpage{1793}
(\byear{1998})
\end{barticle}
\endbibitem

\bibitem{tat_isom}
\begin{barticle}
\bauthor{\bsnm{Tataru}, \binits{D.}}:
\batitle{Rough solutions for the wave maps equation}.
\bjtitle{Amer. J. Math.}
\bvolume{127}(\bissue{2}),
\bfpage{293}--\blpage{377}
(\byear{2005})
\end{barticle}
\endbibitem

\bibitem{krieg_wmcrit}
\begin{barticle}
\bauthor{\bsnm{Krieger}, \binits{J.}}:
\batitle{Global regularity of wave maps from $\mathbb{R}^{2+1}$ to
  $\mathbb{H}^2$}.
\bjtitle{Comm. Math. Phys.}
\bvolume{250}(\bissue{4}),
\bfpage{507}--\blpage{580}
(\byear{2004})
\end{barticle}
\endbibitem

\bibitem{krieg_schlag_ccwm}
\begin{botherref}
\oauthor{\bsnm{Krieger}, \binits{J.}},
\oauthor{\bsnm{Schlag}, \binits{W.}}:
Concentration compactness for critical wave maps.
AMS {M}onographs in {M}athematics
(2012)
\end{botherref}
\endbibitem

\bibitem{sterb_tata_long}
\begin{barticle}
\bauthor{\bsnm{Sterbenz}, \binits{J.}},
\bauthor{\bsnm{Tataru}, \binits{D.}}:
\batitle{Energy dispersed large data wave maps in 2+1 dimensions}.
\bjtitle{Commun. Math. Phys.}
\bvolume{298},
\bfpage{139}--\blpage{230}
(\byear{2010})
\end{barticle}
\endbibitem

\bibitem{sterb_tata_main}
\begin{barticle}
\bauthor{\bsnm{Sterbenz}, \binits{J.}},
\bauthor{\bsnm{Tataru}, \binits{D.}}:
\batitle{Regularity of wave maps in dimension 2+1}.
\bjtitle{Commun. Math. Phys.}
\bvolume{298},
\bfpage{231}--\blpage{264}
(\byear{2010})
\end{barticle}
\endbibitem

\end{thebibliography}


\end{document}